%% file: main.tex
\documentclass[11pt]{article}
\usepackage{graphicx}
\usepackage{amsmath}
\usepackage{amsfonts,bbm}
\usepackage{amssymb}
\usepackage{amsthm}
\usepackage{enumerate,nicefrac,relsize}
\usepackage{fullpage}
\usepackage{hyphenat,epsfig,subcaption,multirow}
\usepackage{nicefrac}
\usepackage{paralist}
\usepackage{pifont}
\usepackage{xargs}
\usepackage{dsfont}
\usepackage[dvipsnames]{xcolor}
\usepackage{tikz}
\usetikzlibrary{shapes.geometric, fit, arrows.meta}
\usepackage{url}

\usepackage[ruled,vlined,linesnumbered]{algorithm2e}



\DeclareFontFamily{U}{mathx}{\hyphenchar\font45}
\DeclareFontShape{U}{mathx}{m}{n}{
      <5> <6> <7> <8> <9> <10>
      <10.95> <12> <14.4> <17.28> <20.74> <24.88>
      mathx10
      }{}
\DeclareSymbolFont{mathx}{U}{mathx}{m}{n}
\DeclareMathSymbol{\bigtimes}{1}{mathx}{"91}

\definecolor{DarkRed}{rgb}{0.5,0.1,0.1}
\definecolor{DarkBlue}{rgb}{0.1,0.1,0.5}

\usepackage[linktocpage=true,pagebackref=true,colorlinks,
  linkcolor=DarkRed,citecolor=ForestGreen,
  bookmarks,bookmarksopen,bookmarksnumbered]{hyperref}
\usepackage[nameinlink,noabbrev]{cleveref}

\usepackage{comment}
\usepackage{tcolorbox}
\usepackage{algorithmic}

\tcbuselibrary{skins,breakable}
\tcbset{enhanced jigsaw}

\usepackage[normalem]{ulem}
\usepackage[compact]{titlesec}

\usepackage{nameref}
\definecolor{ForestGreen}{rgb}{0.1333,0.5451,0.1333}

\usepackage{mdframed}
\usepackage{cite}
\usepackage{enumitem}

\usepackage[margin=1in]{geometry}

\definecolor{Red}{rgb}{0.9,0,0}
\crefname{property}{property}{property}
\crefname{equation}{eq.}{eqs.}

\usepackage{bm}
\usepackage{url}
\usepackage{xspace}
\usepackage[mathscr]{euscript}

\newtheorem{theorem}{Theorem}
\newtheorem{lemma}{Lemma}[section]
\newtheorem{proposition}[lemma]{Proposition}
\newtheorem{corollary}[lemma]{Corollary}

\newtheorem{fact}[lemma]{Fact}

\newtheorem{definition}[lemma]{Definition}

\newtheorem*{claim*}{Claim}
\newtheorem*{theorem*}{Theorem}
\newtheorem*{proposition*}{Proposition}
\newtheorem*{lemma*}{Lemma}
\newtheorem*{problem*}{Problem}

\crefname{lemma}{Lemma}{Lemmas}
\crefname{claim}{Claim}{Claims}
\crefname{fact}{Fact}{Facts}
\Crefname{fact}{Fact}{Facts}

\newtheorem*{mdresult}{Result}
\newenvironment{result}{\begin{mdframed}[backgroundcolor=lightgray!40,topline=false,rightline=false,leftline=false,bottomline=false,innertopmargin=2pt, innerleftmargin=10pt]\begin{mdresult}}{\end{mdresult}\end{mdframed}}

\newtheorem{remark}[lemma]{Remark}
\newtheorem*{remark*}{Remark}

\newtheoremstyle{restate}{}{}{\itshape}{}{\bfseries}{~(restated).}{.5em}{\thmnote{#3}}
\theoremstyle{restate}

\theoremstyle{definition}
\newtheorem{mdalg}{algorithm}

\newtheorem{mddist}{Distribution}

\allowdisplaybreaks

\renewcommand{\qed}{\nobreak \ifvmode \relax \else
      \ifdim\lastskip<1.5em \hskip-\lastskip
      \hskip1.5em plus0em minus0.5em \fi \nobreak
      \vrule height0.75em width0.5em depth0.25em\fi}

\input{macros}

\title{Degree Stability of Parallel Greedy Matching in Random Graphs}

\author{Vishvajeet N\footnote{(nvishvajeet@gmail.com) CNRS, LaBRI, Universit\'{e} de Bordeaux, France.} \and Shrinidhi Teganahally Sridhara\footnote{(shrinidhi.teganahally-sridhara@u-bordeaux.fr) CNRS, LaBRI, Universit\'{e} de Bordeaux, France.}}

\date{}
\begin{document}

\maketitle

\pagenumbering{roman}

\input{abstract}

\clearpage

\setcounter{tocdepth}{3}
\tableofcontents

\clearpage

\pagenumbering{arabic}
\setcounter{page}{1}

\input{introduction}

\input{algorithm}

\input{notation}

\input{warmup}

\input{variance-bound}

\input{linear-rounds}
\input{geometric}

\bibliography{bibliography}

\appendix 

\include{appendix}

\end{document}

%% file: macros.tex
\definecolor{TodoBg}{HTML}{FFF4CC}
\definecolor{TodoFrame}{HTML}{E09F3E}
\definecolor{TodoTitle}{HTML}{9A3412}

\newcommand{\algrule}[1][.2pt]{\par\vskip.5\baselineskip\hrule height #1\par\vskip.5\baselineskip}

\newcommand{\abs}[1]{\left\vert#1\right\vert}

\newcommand{\calF}{\mathcal{F}}
\newcommand{\calH}{\mathcal{H}}
\newcommand{\R}{\mathbb{R}}

\newcommand{\E}{\mathbb{E}}
\newcommand{\Var}{\mathrm{Var}}
\newcommand{\gnp}{\mathbb{G}(n,p)}
\newcommand{\ereg}{\mathcal E_{\mathrm{reg}}}
\newcommand{\estr}{\mathcal E_{\mathrm{struct}}}

\newcommand{\ejump}{\mathcal E_{\mathrm{jump}}}

\newcommand{\bstab}{$\beta$-stability}
\newcommand{\eps}{\epsilon}
\newcommand{\crtcl}{\widetilde{v}}
\newcommand{\diamd}[3]{\xi_{#1}(#2,#3)}
\newcommand{\dtv}{\mathcal D_t(v)}
\newcommand{\mat}[3]{A_{#1}^{#2}(#3)}
\newcommand{\nmat}[3]{B_{#1}^{\neg #2}(#3)}
\NewDocumentCommand{\fil}{m m m O{F}}{%
\mathcal{#4}_{#1}^{#2}%
\if\relax\detokenize{#3}\relax
\else
(#3)%
\fi
}
\newcommand{\filseq}[5]{\Bigl(\mathcal{F}_{#1}^{#2}
\if\relax\detokenize{#3}\relax
\else
(#3)%
\fi
\Bigr)
\if\relax\detokenize{#4}\relax
\else
_{#2 = 0}^{#4}%
\fi
\if\relax\detokenize{#5}\relax
\else
_{#2 \in #5}%
\fi}
\newcommand{\Zmart}[3]{%
\mathbf{Z}_{#1}^{#2}%
\if\relax\detokenize{#3}\relax
\else
(#3)%
\fi
}
\newcommand{\Zmartseq}[5]{\Bigl(\mathbf{Z}_{#1}^{#2}
\if\relax\detokenize{#3}\relax
\else
(#3)%
\fi
\Bigr)
\if\relax\detokenize{#4}\relax
\else
_{#2 = 0}^{#4}%
\fi
\if\relax\detokenize{#5}\relax
\else
_{#2 \in #5}%
\fi}
\newcommand{\Hseq}[2]{\Bigl(\mathcal{H}_{#1} \Bigr)_{#1 \in #2}}
\NewDocumentCommand{\Mart}{m m m O{M}}{%
\mathbf{#4}_{#1}%
\if\relax\detokenize{#2}\relax
\else
^{#2}%
\fi
\if\relax\detokenize{#3}\relax
\else
(#3)%
\fi
}
\newcommand{\tildeMart}[2]{\widetilde{\mathbf{M}}_{#1}(#2)}
\newcommand{\Martseq}[3]{\Bigl(\mathbf{M}_{#1}
\if\relax\detokenize{#2}\relax
\else
(#2)%
\fi\Bigr)_{#1 \in #3}}
\newcommand{\tildeMartseq}[3]{\Bigl(\widetilde{\mathbf{M}}_{#1}(#2) \Bigr)_{#1 \in #3}}
\newcommand{\N}{\mathbb{N}_0}
\newcommand{\erenyi}{Erd\H{o}s--R\'enyi}

%% file: abstract.tex
\maketitle

\begin{abstract}
We study a natural \emph{synchronous} greedy matching process in Erd\H{o}s--R\'enyi random graphs $\mathbb{G}(n,p)$ and give a rigorous analysis of this parallel stochastic dynamics. 
The process evolves as follows: in each round, every unmatched vertex proposes to a uniformly random neighbour, and mutually proposing pairs are matched and removed from the graph, after which the process repeats on the residual graph. We consider the density regime  $p=\Theta(\log^{1+\epsilon} n/n)$ for any constant $\epsilon>0$, which lies just above the threshold for the existence of a perfect matching. We show that, with high probability, for $\Theta(\log^{1+\epsilon} n)$ rounds, every unmatched vertex keeps a constant fraction of its initial degree in the residual graph, simultaneously for all vertices and all rounds. By then, the process has matched a constant fraction of the vertices.

\medskip

Our analysis tracks fine-grained invariants that control the evolution of individual vertex degrees, rather than directly tracking the global degree distribution. At the local level, we analyze the per-round degree loss of a vertex via a neighbourhood-exposure martingale. We show that  structural sparsity of $\mathbb{G}(n,p)$ yields $O(1)$ variance and thus the per-round degree loss is strongly concentrated around its conditional expectation. At the global level, these local concentration bounds control the increments of a martingale tracking cumulative degree loss over many rounds, allowing us to show that vertex degrees remain stable for $\Omega(\log^{1+\epsilon} n)$ rounds.
\end{abstract}

%% file: introduction.tex

\section{Introduction}\label{sec:intro}

Matchings are one of the most fundamental structures in graph theory and combinatorial optimization, and the problem of finding large matchings has been extensively studied across a range of models, in both worst-case and random graph settings. These two perspectives lead to fundamentally different algorithmic and analytical challenges.

\smallskip

Despite decades of work following Edmonds' blossom algorithm~\cite{Edmonds65}, no algorithm asymptotically faster than the Micali--Vazirani algorithm~\cite{MicaliV80,Vazirani24} is known for maximum matching in general graphs.
Among the simplest approaches for matchings are greedy algorithms, which iteratively select edges according to simple local rules. While they have poor worst-case guarantees, greedy algorithms often perform remarkably well on random inputs and frequently serve as primitives in large-scale
and distributed systems.

\smallskip

The seminal work of Karp and Sipser~\cite{KarpS81} initiated the study of sequential greedy processes for matchings on random graphs. A large body of subsequent work shows that these processes produce matchings close to maximum with high probability and provides a detailed understanding of their output distributions (see, e.g.,~\cite{AronsonFP98,Kreacic17,GlasgowKSS25} and references therein).
These analyses are primarily carried out using the differential equation method~\cite{Wormald95}, which tracks the gradual evolution of the degree distribution of the graph when edges are selected one at a time.

\smallskip

In many modern algorithmic settings, greedy algorithms are executed in parallel without centralized coordination, allowing many vertices to act simultaneously and significantly accelerate the process. At the same time, parallelism complicates the analysis: in a single round, many vertices are
matched at once, which can remove a linear number of edges, and the decisions made in that
round are correlated through overlapping neighbourhoods. An analysis therefore has to control
many simultaneous, dependent changes, rather than one edge at a time.

\smallskip

We consider a simple synchronous process for constructing matchings. In each round, every unmatched vertex proposes to a uniformly random neighbour, and pairs of vertices that mutually propose to each other are matched and removed from the graph. The process is then repeated on the residual graph. 

\smallskip

We analyze these dynamics on  \erenyi~random graphs $G\sim\gnp$ in the regime
$p = \Theta\!\left(\frac{\log^{1+\epsilon} n}{n}\right)$
for any constant $\eps>0$, which lies just above the threshold for the existence of perfect matchings.

\begin{result}[Informal]
Fix any constant $\epsilon>0$ and consider  $G \sim \gnp$ with $p=\Theta(\log^{1+\epsilon} n/n)$. During an execution of the parallel greedy process for $\Theta(\log^{1+\epsilon} n)$ rounds, with high probability, every unmatched vertex retains at least a constant fraction of its initial degree, simultaneously for all vertices and for all rounds.
Consequently, the process matches a linear number of vertices during these rounds.

\end{result}

Formally, the robustness property is captured in our main \Cref{thm:persistence}: although vertices are removed as matches occur, every remaining vertex retains many available edges.
This prevents the algorithm from quickly ``running out of edges'' and allows it to make steady progress over the analyzed rounds, as is formalized in \Cref{thm:linear-rounds-matching}.

Our analysis tracks the local evolution of vertex degrees rather than the global degree distribution. The key idea is a local-to-global martingale argument: we first obtain strong control on the degree loss of a vertex within each round, and then use this to bound the cumulative deviation of degree loss over many rounds. This prevents local fluctuations from building up, ensuring that degrees remain stable throughout the process for $\Omega(\log^{1+\epsilon} n)$ rounds.

\medskip
\subsection{Proof Overview}

We begin with an overview of the main ideas behind the proof. At a high level, the analysis proceeds from local to global: we first obtain strong control over the degree loss of individual vertices within a single round, which yields bounded increments for a martingale tracking cumulative degree loss, and then use this to control the cumulative deviation over many rounds.

\smallskip

Recall our matching process: on a random graph $G \sim \gnp$, in each round every unmatched
vertex proposes to a uniformly random neighbour, and two vertices are
matched if and only if they propose to each other.
We emphasize that the randomness governing the proposal choices is independent across vertices and across rounds.
This independence plays a crucial role in the analysis by partially
decorrelating the local interactions created by the simultaneous
proposals.

\paragraph{A guiding invariant for large matchings: degree-positivity.}

We begin with a simple observation that guides our entire analysis:
the growth of the matching can be reduced to maintaining a local
structural invariant.
We observe that if every unmatched vertex has available at least one unmatched neighbour at the
start of each round -- a property we call 
\emph{degree-positivity invariant} -- then every vertex has a
non-trivial probability of becoming matched in that round. In particular,  suppose degree-positivity holds, then we have for any vertex $v$,
\begin{align*}
	\label{eq:match}
	\Pr\!\left(v \text{ is matched} \right)
	=
	\sum_{u\in N(v)}
	\frac{1}{\deg (v)}\cdot 	\frac{1}{\deg (u)},
\end{align*}
where $\deg(\cdot)$ denotes the degree of a vertex and $N(v)$ denotes the set of neighbours of $v$.

\smallskip

Indeed, conditioning on the current residual graph, if the maximum
degree of the input graph is at most $\Delta$, then every unmatched vertex becomes matched
with probability at least $1/\Delta$. Thus, in each round, every unmatched vertex has a large probability of being removed. It follows that the probability that a fixed vertex remains unmatched for $T$ rounds decays exponentially in $T$, as long as degree-positivity is maintained. This connection is formalized in \Cref{lem:geometric-decay}.

\smallskip

Random graphs $\gnp$ with $p = c\log^{1+\epsilon}n/n$, for some constant $c>0$ and any constant $\epsilon>0$, have tightly concentrated degrees, so that $\Delta=\Theta(d_0)$, where $d_0=\log^{1+\epsilon}n$ denotes the order of the expected degree. 

Our goal is therefore to understand how long the degree-positivity
invariant can persist during the execution of the process:

\bigskip
\centerline{\emph{For how many rounds can the degree-positivity invariant be
maintained with high probability?}}
\medskip
In fact, we prove more: during these rounds every unmatched vertex keeps a constant fraction of
its initial degree. This stronger property, which we call \bstab, is what our analysis carries
from one round to the next.

\paragraph{How could the invariant fail?}

We begin by examining how the invariant could fail.
Consider an unmatched vertex $v$ at some round $t$.
The only way $v$ can violate the invariant is if it loses all of its
neighbours before it becomes matched itself.
This happens when all the neighbours of $v$ become matched elsewhere during
the process.

\smallskip

At first sight it seems possible that a vertex $v$ could lose all of its
neighbours before being matched. However, if the events that
different neighbours of $v$ become matched were independent, then the
probability that all of them disappear before $v$ itself is matched
would be extremely small.

\smallskip

The difficulty is that these events are not independent: matching
decisions depend on overlapping local neighbourhoods. Consequently,
correlations between the matching events of neighbours of $v$ could in
principle allow neighbour losses to accumulate over many rounds and
eventually isolate $v$.

\paragraph{Understanding the first round.}

To understand this phenomenon of isolation without being matched, we
begin with a warm-up analysis of the first round.

For a fixed vertex $v$, let
\[
X_1(v) := |\{u\in N(v): u \text{ is matched in round }1\}|
\]
denote the number of neighbours of $v$ that become matched in the first
round.

A simple calculation shows that
\[
\mathbb{E}[X_1(v)] \le \frac{\Delta_0}{\delta_0},
\]
where $\Delta_0$ and $\delta_0$ denote the maximum and minimum degrees
of the graph.

As the degrees in $\gnp$ are tightly concentrated, this ratio is a
constant, implying that the expected number of neighbours lost by $v$
in a round is small compared to its degree.

\smallskip

We next ask whether $X_1(v)$ is tightly concentrated around its mean.
To show this, we expose the proposal choices of the vertices in the neighbourhood of $v$ sequentially and analyze
their effect on $X_1(v)$ using a martingale.

\smallskip

Let $N(v) = \{u_1, \dots, u_d\}$ denote the neighbours of $v$ in the initial graph.
We reveal the random proposals $\pi(u_1), \dots, \pi(u_d)$ one by one, where $\pi(u)$ denotes the neighbour to which vertex $u$ proposes in the round.
Let
$\fil{1}{k}{v}$ denote the information revealed after exposing the first $k$
proposals. We then consider the Doob martingale
\[
\Zmart{1}{k}{v} := \mathbb{E}[X_1(v)\mid \fil{1}{k}{v}], \qquad k=0,1,\dots,d,
\]
which we refer to as the \emph{neighbourhood-exposure martingale}.

\smallskip

Revealing the proposal of a single neighbour $u_k$ can only affect matching
decisions involving $u_k$. Since each vertex can be part of at most one matched
edge, changing $\pi(u_k)$ can create or destroy at most one matched edge incident
to $u_k$. If this edge connects $u_k$ to another neighbour of $v$, then the matched
status of both endpoints may change, and hence $X_1(v)$ can change by at most $2$.

\smallskip

Thus, the martingale $\Zmartseq{1}{k}{v}{}{\N}$ satisfies the bounded difference property,
$|\Zmart{1}{k}{v} - \Zmart{1}{k-1}{v}| \le 2$ for $k \in \N$. Since $X_1(v)$ depends on $d = \Theta(d_0)$ independent
choices, applying Azuma--Hoeffding inequality yields
\[
\Pr(|X_1(v)-\mathbb{E}[X_1(v)]|\ge \lambda)
\le
\exp\!\left(-\Theta\!\left(\frac{\lambda^2}{d_0}\right)\right).
\]
In particular, deviations of order $O(\sqrt{d_0\log n})$ already occur
with inverse polynomially small probability:
\[
\Pr\!\left(
|X_1(v)-\mathbb{E}[X_1(v)]| \ge c\sqrt{d_0\log n}
\right)
\le O(n^{-10}) .
\]
Thus, by a union bound, we have the guarantee that with high probability no vertex
loses more than $O(\sqrt{d_0\log n})$ neighbours in a single round.
However, these deviations are still too large to allow this argument to
be applied iteratively over many rounds. Even though the expected
degree loss in each round is small, 
fluctuations of order
$O(\sqrt{d_0\log n})$ could accumulate across rounds and eventually
cause a vertex to lose a large fraction of its neighbours. As a result,
this approach only provides sufficient control to maintain the
invariant for $O\left(\sqrt{\frac{d_0}{\log{n}}}~\right)=
O (\log^{\epsilon/2}n)
$ rounds. 

\paragraph{Strong local concentration using the structure of $\gnp$.}
Degree positivity cannot be carried from one round to the next on its own. Whether $v$ keeps a
neighbour in round $t$ depends on how fast its neighbours are matched, and the probability that a
neighbour $u$ is matched depends on the degrees of the neighbours of $u$, about which degree
positivity says nothing. We therefore work under the stronger invariant we call \bstab~(\Cref{def:stability-invariant}): for a constant $\beta>0$, $every$ unmatched vertex keeps at least
$\beta d_0$ neighbours in $every$ round considered,
\[
\deg_t(v)\ge\beta d_0\qquad\text{for all } v\in U_t .
\]
Under this invariant each neighbour of $v$ is matched with probability $\Theta(1/d_0)$ in a round, so the
expected degree loss satisfies $\E[X_t(v)]\leq \Theta(1)$, and degree positivity holds automatically.

To analyze the degree loss under this invariant, we fix a vertex $v$ that is unmatched at the start of round $t$, and write 
\[
X_t(v)=\sum_{u\in N_{t-1}(v)} I_t(u).
\]
Here $I_t(u)$ is the indicator of the event that vertex $u$ becomes matched in round $t$ and $N_{t-1}$ is the set of neighbours of $v$ in the residual graph at the start of round $t$.

\smallskip

While the variables $\{I_t(u)\}$ are not independent, their dependencies are
extremely limited for random graphs in our regime.
In particular, overlaps between the neighbourhoods of vertices in
$N_{t-1}(v)$ are rare, and even when they occur, they involve only a very
small number of vertices.

\smallskip

We show that the dependency graph among the indicators $\{I_t(u)\}$
has maximum degree at most $1$, and is therefore a matching. This allows us to decompose the neighbourhood of $v$ into disjoint blocks of size at most $2$ that account for the dependencies.

\smallskip

More precisely, this induces a partition of $N_{t-1}(v)$ into disjoint blocks
$B_1,\dots,B_m$, and we can write
\[
X_t(v) = \sum_{i=1}^m Y_i,
\]
where each $Y_i$ is the sum of at most two dependent indicators.

\smallskip

To analyze this sum, we first condition on the proposal $\pi_t(v)$, where $\pi_t(v)$ denotes the neighbour to which vertex $v$ proposes in round $t$. This conditioning removes the only global source of dependence among the variables.
We then expose the blocks sequentially and consider the corresponding
Doob martingale. We show that each increment is bounded by a constant and that, under the
$\beta$-stability invariant, each block contributes only $O(1/d_0)$ to
the conditional variance. Summing over all blocks then yields that the
total conditional variance in a round is $O(1)$

\smallskip

Applying Freedman's inequality (see \Cref{fact:freedman} for a precise statement) to this martingale, together with the
bounded increments and $O(1)$ quadratic variation, yields a strong
concentration bound (we formalize this in \Cref{lem:bernstein-inround}):
\[
\Pr(|X_t(v)-\mathbb{E}[X_t(v)]|\ge \lambda )
\le
\exp\!\left(-\Theta\!\left(\frac{\lambda^2}{1+\lambda}\right)\right).
\]

In particular, deviations of order $O(\log n)$ already occur with
polynomially small probability:
\[
\Pr\!\left(
|X_t(v)-\mathbb{E}[X_t(v)]|
\ge c\log n
\right)
\le
O(n^{-10}).
\]
Compared to the Azuma--Hoeffding bound obtained previously, this improves the deviation scale
from $O(\sqrt{d_0\log n})$ to $O(\log n)$ and yields much stronger tail probabilities.

We now show how these local deviation bounds can be combined across
many rounds to control the global evolution of vertex degrees throughout the
process.

\paragraph{Controlling degree-loss over many rounds by cumulative-deviation martingales.}

Our local Freedman-based bounds show that the degree loss in a single round is tightly concentrated around its expectation. However, they do not by themselves rule out the possibility that small deviations accumulate over many rounds. In particular, even with these stronger bounds, a na\"{\i}ve union bound over rounds only ensures that the invariant persists for $O(d_0/\log n)=O(\log^\epsilon n)$ rounds, since each round may incur a loss of up to $O(\log n)$ neighbours. 

To overcome this limitation, we track how these deviations accumulate over time via a cumulative-deviation martingale. Fix a vertex $v$, and define
\[
\Mart{0}{}{v}=0,
\qquad
\Mart{t}{}{v}=\Mart{t-1}{}{v}+
\Bigl(
X_t(v)-\mathbb{E}[X_t(v)\mid\mathcal H_{t-1}]
\Bigr).
\]

Here $\mathcal H_{t-1}$ denotes the history of the process up to the start
of round $t$. By construction $\Martseq{t}{v}{\N}$ is a martingale with respect to the filtration $\big(\mathcal{H}_t\big)_{t\in \N}$.
A na\"{i}ve attempt to control this martingale would bound the increments
$X_t(v)$ deterministically.
Since a vertex may have degree as large as $O(d_0)$, this would yield
a worst-case bound of $O(d_0)$ on the martingale increments.
Such coarse bounds are too weak to control the cumulative deviation
over  $\Theta(d_0)$ rounds.

\smallskip

Our key idea is instead to use the stronger probabilistic control
obtained from the local concentration bounds established earlier.
Those results show that with extremely high probability the
per-round degree loss satisfies
\[
X_t(v)=O(\log n).
\]
Thus, rather than treating the martingale increments as potentially as
large as $d_0$, we may treat them as typically logarithmic.

\smallskip

More precisely, we proceed via a two-level analysis. \emph{First}, the
Freedman-based bounds obtained above provide strong control over the
\emph{local} degree loss in each round, showing that the increments of
the cumulative-deviation martingale are typically only $O(\log n)$.
\emph{Second}, we control the accumulation of these local deviations
across rounds via the martingale $\Mart{t}{}{v}$.

\smallskip

Our local concentration bounds ensure that the martingale increments are
$O(\log n)$ with high probability. Crucially, although the increments can be as large as $O(\log n)$, such large deviations occur with very small probability. Under the $\beta$-stability invariant, the conditional variance of each increment is $O(1)$, so over $T=\Theta(d_0)$ rounds the predictable quadratic variation is $O(d_0)$. Thus, Freedman's inequality applies in a regime where the total variance is small, yielding global control on the cumulative deviation $\Mart{T}{}{v}$ and showing that the cumulative deviation is $O(d_0)$ rather than $O(d_0 \log n)$. The logarithmic bound on the increments is essential here: using only the trivial bound $O(d_0)$ would yield a constant-probability deviation bound from Freedman's inequality, which is insufficient for our purposes, whereas the sharper $O(\log n)$ bound ensures exponentially small tail probabilities.

\smallskip

To make this argument rigorous, we use a stopping-time construction. We
define a stopping time corresponding to the first round in which some
vertex violates the $\beta$-stability invariant. Up to this time, the
degree of every vertex remains bounded below by a constant fraction of
its initial value, ensuring that the conditions required for the above
concentration and martingale bounds continue to hold.

\smallskip

Applying Freedman's inequality up to this stopping time shows that,
with high probability, it occurs only after $\Omega(d_0)$ rounds.
Consequently, the $\beta$-stability invariant holds for $\Omega(d_0)$
rounds (formalized in \Cref{thm:persistence}).

\smallskip

Since $\beta$-stability implies the degree-positivity invariant holds, the
geometric decay argument then guarantees that the algorithm produces a
matching of size $\Theta(n)$ (this is formalized in \Cref{thm:linear-rounds-matching}). 

Our work studies the greedy parallel dynamics at our chosen density $p=\frac{\log ^{1+\epsilon}n}{n}$, which is just above the the perfect matching threshold $p=\frac{\log n + O(\log\log n)}{n}$. This leaves the general question open:

\medskip
\centerline{\emph{What is the behaviour of the parallel greedy process for other densities $p$?}}
\medskip

Additionally, we are interested in the `end-game', i.e. to understand the matching process and its trajectory after it has matched a constant fraction of the vertices, at all densities $p$.

\medskip
\subsection{Related Work}\label{sec:related-work}
The design of efficient matching algorithms is a central topic in combinatorial optimization. Following Edmonds' landmark blossom algorithm \cite{Edmonds65} and the $O(\sqrt{|V|}\cdot |E|)$ algorithm of Micali and Vazirani~\cite{MicaliV80} (whose correctness was only recently formalized in \cite{Vazirani24}), substantial research has focused on improving worst-case performance. Despite decades of work, no asymptotically faster algorithm than that of Micali and Vazirani is known for maximum matching in general graphs.

\smallskip

On the complexity side, parallel matching has been extensively studied. 
Perfect matching has long been known to lie in $\mathcal{RNC}$ via randomized algebraic methods~\cite{MulmuleyVV87, KarpUW88}, most notably through the use of the isolation lemma. 
This connection has driven a substantial line of research on constructing deterministic or pseudorandom weight assignments that isolate a unique perfect matching. 
A recent breakthrough by Fenner, Gurjar, and Thierauf~\cite{FennerGT16} gives a near-derandomization of the isolation lemma, placing bipartite perfect matching in quasi-$\mathcal{NC}$. Subsequently, this result was extended to general graphs by Svensson and Tarnawski \cite{SvenssonT17}. On the algorithmic side, distributed and parallel algorithms for matching have been extensively studied. A classical result of Israeli and Itai~\cite{IsraeliI86} gives a randomized algorithm that computes a maximal matching (yielding a $1/2$-approximation to the maximum matching) in $O(\log n)$ synchronous rounds on general graphs. Subsequent work has focused on improving approximation guarantees and extending to weighted settings, including $(1-\varepsilon)$-approximation algorithms in the CONGEST model~\cite{LotkerPP15}, as well as earlier constant-factor approximation algorithms for weighted matching~\cite{LotkerPR07, WattenhoferW04, Hoepman04, CzygrinowHS04}.

One round of our process is the randomized handshake (rendezvous) procedure of M\'etivier, Saheb and
Zemmari~\cite{MetivierSZ03}: each vertex chooses a uniformly random neighbour, and two vertices that
choose each other are matched. Its one-round behaviour has been analysed on general
graphs~\cite{MetivierSZ03,Dietzfelbinger02,DietzfelbingerT05} and on $\gnp$~\cite{Zemmari08}, and its
iteration on the unmatched vertices, which is the process studied here, terminates with
probability~$1$~\cite{FontaineZ16}; the expected number of rounds of a closely related variant has also
been bounded~\cite[Ch.~6]{OuledAbdallah17}. The residual degrees, which are the subject of this paper,
are not analysed in these works.

\smallskip

We note that our argument follows every vertex round by round, in the style of the semi-random, or ``nibble'',
method of R\"odl~\cite{Rodl85} and Pippenger and Spencer~\cite{PippengerS89}. This method has also been applied to processes such as edge coloring, triangle-free and others~\cite{DubhashiGP98,GrableP97,BohmanK21,BennettB19}.

The key idea behind our analysis is to control the degree loss via a two-level martingale process. One that proves concentration around expected degree loss in a round and one that proves that the accumulation of the degree loss over multiple rounds is controlled.

\smallskip

While classical matching algorithms rely on augmenting path techniques,
greedy algorithms offer simple and efficient approaches to constructing matchings. However, in the worst case they achieve only constant-factor approximations. 
Understanding the limitations of randomized greedy matching algorithms in the worst case has been an important direction~\cite{KorteH78, DyerF91,PoloczekS12}. 
This has motivated the study of greedy algorithms under random graph models, where such worst-case lower bounds are not known.

\smallskip

Classical results identify sharp thresholds for the existence of perfect matchings in \erenyi~random graphs: above this threshold, such matchings exist with high probability \cite{ErdosR66, BollobasF85}.  Moreover, randomized algorithms can find them efficiently; for example, Angluin and Valiant \cite{AngluinV77} gave an $O(n \log n)$ algorithm that succeeds w.h.p. in graphs with at least $c n \log n$ edges.

\smallskip

An extensive line of work, initiated by the seminal paper of Karp and Sipser~\cite{KarpS81}, studies the performance of  greedy processes for finding large matchings in even sparser regimes. Specifically, Karp and Sipser analysed a simple greedy algorithm and showed that it finds a matching which is within $o(n)$ of the maximum
in $\gnp$ with $p= \frac cn$ where $c > 0$ is a constant. 
Subsequent refinements sharpened this analysis significantly. In particular, Aronson, Frieze, and Pittel~\cite{AronsonFP98} showed that the Karp--Sipser algorithm produces a matching within $\widetilde O(n^{1/5})$ of the maximum with high probability. Related work further established precise asymptotic and distributional properties of greedy matching algorithms on random graphs~\cite{DyerFP93, AronsonDFS95}. Complementary algorithmic results of Chebolu, Frieze, and Melsted~\cite{CheboluFM10} show that such near-maximum matchings can be augmented to a maximum matching in $O(n)$ expected time. More recent works, including the PhD thesis of Krea\v{c}i\'c~\cite{Kreacic17} and Glasgow, Kwan, Sah, and Sawhney~\cite{GlasgowKSS25}, establish central limit theorems for the matching size in sparse random graphs, building on the analysis of the Karp--Sipser process and showing that its fluctuations are asymptotically Gaussian.

\smallskip

Such analyses have also been extended beyond the general Erd\H{o}s--R\'enyi model: e.g., Bohman and Frieze~\cite{BohmanF11} analyze the Karp--Sipser algorithm on random graphs with a fixed degree sequence, showing that even under strong structural constraints, the algorithm finds almost-perfect matchings.

This line has helped establish the differential equation method (DEM)~\cite{Wormald95} as a powerful framework for analyzing stochastic processes on random graphs.

\smallskip

We study a fully synchronous process in which all vertices act in parallel each round, a natural model for greedy matching dynamics. While algorithms for matching have been studied extensively in parallel and distributed settings, their focus is on worst-case guarantees -- such as achieving constant-factor approximations in a small number of communication rounds -- on arbitrary graphs. In contrast, our focus is on the stochastic evolution of fully symmetric parallel greedy dynamics on random graphs, where the goal is to understand the trajectory of the process rather than to optimize worst-case approximation ratios.

\smallskip

\medskip
\subsection{Organization}\label{sec:organization}

We describe the randomized greedy matching process in
\Cref{sec:algorithm}, while \Cref{sec:prelim} introduces the notation,
random graph model, and concentration tools used throughout the paper.

\smallskip

\Cref{sec:warmup} serves as a warmup, analyzing the initial rounds via a neighbourhood-exposure martingale and the Azuma--Hoeffding inequality to obtain weak concentration bounds. This shows that degree-positivity persists for the first $\Theta(\log^{\epsilon/2} n)$ rounds.

\smallskip

\Cref{sec:strong-conc} exploits structural sparsity of $\gnp$ to obtain
strong local concentration bounds; the dependency graph among matching
indicators is a matching, allowing us to decompose dependencies into
small blocks and show that the variance of the degree loss in each
round is $O(1)$.

\smallskip

\Cref{sec:deg-evolution} combines these local bounds with a cumulative
martingale argument to control how degree fluctuations accumulate over
$\Theta(\log ^{1+\eps} n)$ rounds. Using Freedman's inequality together with a
stopping-time argument, we show that the $\beta$-stability invariant
persists for $\Omega(\log ^{1+\eps} n)$ rounds (\Cref{thm:persistence}),
which is the main focus of this work.

\smallskip

In \Cref{sec:deg-positivity} we show that if the degree-positivity invariant holds for $T$ rounds,
then the number of unmatched vertices decays geometrically (\Cref{lem:geometric-decay}), thus giving large matchings (\Cref{thm:linear-rounds-matching}).

\smallskip

Additional concentration results
are collected in the Appendix.

%% file: algorithm.tex
 \newpage
\section{A Simple Synchronous Greedy Process for  Matchings}\label{sec:algorithm}

We study a simple randomized greedy process that constructs a matching in an input graph
$G(V,E)$ over a sequence of synchronous rounds. The process is vertex-driven: in each round, each vertex
\emph{independently} proposes to a \emph{uniformly random} neighbour. 

\smallskip

We describe the process below.

\crefname{algorithm}{}{ALG}
\Crefname{algorithm}{\!\!}{ALG}
\begin{algorithm}[H]
\caption{Randomized Greedy Matching}
\label{alg:main}
\DontPrintSemicolon
\BlankLine 
\SetKwInOut{Input}{Input}
\SetKwInOut{Output}{Output}
\SetAlgoNoLine
\SetAlgoNoEnd
\LinesNotNumbered 
\Input{Graph $G(V,E)$, number of rounds $T$}
\Output{Matching $M \subseteq E$}

\algrule[.55pt]
 $M \leftarrow \emptyset$, $U \leftarrow V$\;

\For{$t = 1$ \KwTo $T$}{
    
    \ForEach{$v \in U$ s.t. $N_G(v) \cap U\neq \emptyset $ \textbf{independently}}{
        Choose $\pi(v)$ uniformly at random
         from $N_G(v) \cap U$\;
    }

    \ForEach{$(u,v)$ s.t. $u\in U$ and $v \in N_G(u) \cap U$}{
        \If{$\pi(u) = v$ and $\pi(v) = u$}{
            $M \leftarrow M \cup \{(u,v)\}$\;
            $U \leftarrow U \setminus \{u,v\}$\;
        }
    }
}

\Return{$M$}\;

\end{algorithm}

The parameter $T$ specifies the number of rounds for which the process
is executed. Although the process eventually converges to a maximal
matching as $T\to\infty$, our analysis concerns its first $T$ rounds.

\smallskip

We observe that \Cref{alg:main} always outputs a valid matching
of the input graph.

\begin{lemma}[Correctness of \ref{alg:main}]
The set $M$ output by \ref{alg:main} is a matching of the input graph $G(V,E)$.
\end{lemma}
\begin{proof}
We show that no vertex is incident to more than one edge in $M$.
In each round, every vertex $v\in U$ selects exactly one neighbour
$\pi(v)\in N_G(v)\cap U$.
An edge $(u,v)$ is added to $M$ only if $\pi(u)=v$ and $\pi(v)=u$,
and only while both $u$ and $v$ belong to $U$.
Once such an edge is added, both vertices are removed from $U$
and do not participate in subsequent rounds.
Therefore, a vertex can be matched in at most one round and cannot
be incident to more than one edge of $M$.
Hence $M$ is a matching of $G$.
\end{proof}

The remainder of the paper analyzes the behaviour of \Cref{alg:main} on
\erenyi~random graphs.

%% file: notation.tex
\newpage
\section{Notation and Preliminaries}\label{sec:prelim}

In this section we introduce the notation and probabilistic tools used
throughout the analysis.

\paragraph{Matching process notation.}
We begin by introducing notation for the evolution of the matching
process. For each $t \ge 0$, let $M_t$ denote the matching constructed
by the end of round $t$ during the execution of
\Cref{alg:main}, and let $U_t$ denote the set of vertices that remain
unmatched at the end of round $t$. We set $M_0=\emptyset$ and
$U_0=V(G)$, and refer to the initial configuration before any proposals
are made as round $0$.
\smallskip

For each $t \ge 0$ we define the \emph{residual graph} $G_t$
as
\[
G_t := G[U_t].
\]
For any $v \in U_t$, let $\deg_t(v)$ denote the degree of $v$ in $G_t$,
and let $N_t(v)$ denote its open neighbourhood in $G_t$.
\smallskip

Finally, let
\[
\Delta_t := \max_{v \in U_t} \deg_t(v),
\qquad
\delta_t := \min_{v \in U_t} \deg_t(v)
\]
denote the maximum and minimum degrees respectively, in $G_t$.
\smallskip

To track the evolution of the process round by round, we introduce the
random variables $I_t(u)$ and $X_t(v)$ defined below. The variables $I_t(u)$ capture which vertices become matched in round $t$, while $X_t(v)$ tracks the number of neighbours of $v$ that disappear in that round.
\begin{definition}[Matching indicator $I_t$]
Fix a round $t \ge 1$ and a vertex $u \in U_{t-1}$. We define the indicator random variable $I_t(u)$ as,
\[
I_t(u) := \mathds{1}\{u \in U_{t-1} \setminus U_t\}.
\]
\end{definition}

\begin{definition}[Round-$t$ degree-loss]
\label{def:xvt}
Fix a round $t \ge 1$ and a vertex $v \in V(G)$. Define the random variable
\[
X_t(v)
:=
\begin{cases}
\displaystyle\sum_{u \in N_{t-1}(v)} I_t(u),
& \text{if } v \in U_{t-1}, \\[6pt]
\hspace{21mm} 0, & \text{if } v \notin U_{t-1}.
\end{cases}
\]
\end{definition}

Although $X_t(v)$ primarily tracks the degree loss of unmatched
vertices, we extend its definition to the entire vertex set $V(G)$ by
setting $X_t(v)=0$ whenever $v \notin U_{t-1}$.
This allows us to analyze
the algorithm as a stochastic process evolving on the \emph{fixed index
set} $V(G)$, significantly simplifying the tracking of degree-loss over
multiple rounds.
\smallskip

We observe that for every $v \in U_{t}$ (that is, $v$ is still unmatched after round $t$),
\[
\deg_t(v)=\deg_{t-1}(v)-X_t(v).
\]
\paragraph{Invariants.} We now define the key invariants used in this work. The degree-positivity invariant ensures that every vertex has at least one neighbour to match with at the start of each round when it holds, while \bstab~ensures that the degree of every vertex in the residual graph remains at least a constant fraction of its initial degree.
\begin{definition}[Degree-positivity invariant up to round $T$]
\label{def:degree-positivity}
We say that the degree-positivity invariant holds up to round $T$ if:

$\text{For every } t \leq T \text{ and every } v \in U_{t-1}$, we have
\[
\deg_{t-1}(v) \ge 1.
\quad
\]
\end{definition}
\begin{definition}[\bstab~invariant]
    \label{def:stability-invariant}
    Let $\beta >0$ be a constant. We say that the $\beta$-stability invariant holds up to round $T$ if:
    
    $\text{For every } t \leq T \text{ and every } v \in U_{t-1}$, we have
    \[
        \deg_{t-1}(v) \geq \beta d_0.
    \]
    Equivalently, 
    \[
        \delta_{t-1} \geq \beta d_0.
    \]
\end{definition}

We note that if the $\beta$-stability invariant holds up to round $t$,
then the degree-positivity invariant also holds up to round $t$.

\paragraph{Random graph model.}
Throughout the paper we consider graphs drawn from the Erd\H{o}s--R\'enyi
random graph distribution defined below. 
\begin{definition}[Erd\H{o}s--R\'enyi random graph distribution]
\label{def:gnp}
For an integer $n \ge 1$ and a real number $p \in [0,1]$, the distribution $\gnp$ is the probability distribution over graphs on
vertex set $[n]$ in which each edge $\{u,v\}$ with $u \neq v$
is included independently with probability $p$.
\end{definition}

We write $G \sim \gnp$ for a graph sampled from this distribution,
and occasionally abuse notation by writing $\gnp$ to denote the same.
\smallskip

\paragraph{Density regime.}
Throughout the paper, we work with random graphs $G \sim \gnp$ with
\[
p = \frac{c\log^{1+\epsilon} n}{n}
\]
for \emph{any} $\epsilon>0$ and a constant $c>0$ sufficiently large. This regime lies just above the threshold for perfect matchings, which occurs at $p = \frac{\log{n} + O(\log\log{n})}{n}$
\cite[Theorem 6.2]{FriezeK15}. 
Some of our structural lemmas already hold under the weaker assumption $\epsilon\geq 0$ which includes the regime $p= c\log n/n$, while our main theorem requires the regime $\epsilon>0$.

\paragraph{Typical degree scale.}
For convenience we denote by
\[
d_0 := \log^{1+\epsilon} n
\]
the canonical degree scale in the density regime considered above. In
particular, when $G \sim \gnp$ with $p = \Theta(\log^{1+\epsilon}n/n)$,
the expected degree satisfies $pn = \Theta(d_0)$. We therefore use $d_0$
as a convenient shorthand for the typical magnitude of vertex degrees
in the graph.
\smallskip

A basic property of Erd\H{o}s--R\'enyi random graphs is that vertex degrees are concentrated around their expectation. We define the  event called degree-regularity which captures this quantifiably,
\begin{definition}[Degree-regularity event ($\mathcal E_{\mathrm{reg}}$)]
\label{def:ereg}
Let $G \sim \gnp$. We define the event $\mathcal E_{\mathrm{reg}}$ that the initial degrees are concentrated around the mean; specifically, 
\[
\frac{pn}{2} \le \deg_G(v) \le 2pn
\qquad
\text{for all } v \in V(G).
\]
\end{definition}
We require throughout this article a precise bound on the probability that $\ereg$ holds, and we prove one in \Cref{fact:deg-conc}.
\smallskip

For a random variable $X$, we write $\E[X]$ for its expectation and
$\mathrm{Var}(X) := \E[(X-\E[X])^2]$ for its variance.

\paragraph{Martingale preliminaries.}
Our analysis relies on martingale techniques, and we briefly recall the
definitions we will use. We use bold letter font to denote random variables concerning martingales. More details and background on martingales can be found in any standard textbook on probability theory, e.g. \cite{GrimmettS01}.

\begin{definition}[Filtration]
Let $(\Omega,\fil{}{}{}, \Pr)$ be a probability space.
A sequence of $\sigma$-algebras $\big(\mathcal{F}_t\big)_{t\in \N}$ is called a
\emph{filtration} if
\[
\fil{0}{}{} \subseteq \fil{1}{}{} \subseteq \fil{2}{}{} \subseteq \cdots .
\]
\end{definition}
Throughout the paper we write ``conditioning on $\calF$'' to mean
conditioning on an arbitrary realization of the information contained
in the $\sigma$-algebra $\calF$,  equivalently, conditioning on an arbitrary \emph{event} in $\calF$.

\begin{definition}[Martingale]
Let $\big(\mathcal{F}_t\big)_{t\in \N}$ be a filtration.
A sequence of random variables $\big(\Zmart{t}{}{}\big)_{t\in \N}$ is called a
\emph{martingale with respect to the filtration $\big(\mathcal{F}_t\big)_{t\in \N}$} if
\[
\E\left[\abs{\Zmart{t}{}{}}\right] < \infty
\quad \text{for all } t \in \N,
\]
and for every $t \ge 1$,
\[
\E\left[\Zmart{t}{}{} \mid \fil{t-1}{}{}\right] = \Zmart{t-1}{}{}.
\]
\end{definition}

\begin{definition}[Doob martingale]
Let $\big(\mathcal{F}_t\big)_{t\in \N}$ be a filtration and let $X$ be an integrable random variable with $\E[X] < \infty$. The sequence
\[
\Zmart{t}{}{} := \E\left[X \mid \fil{t}{}{}\right], \qquad t \ge 0,
\]
is called the \emph{Doob martingale} associated with $X$ and the filtration $\big(\mathcal{F}_t\big)_{t\in \N}$.
\end{definition}

We will use $[n]$ to denote the set $\{1,2, \ldots n\}$, and $\mathbb N_0$ denotes the set $\mathbb N \cup \{0\}$, where $\mathbb N$ is the set of natural numbers.  
\smallskip

\paragraph{Matching probabilities under \Cref{alg:main}}
To analyze the process round by round, we let $\mathcal H_t$ denote the history of the algorithm up to the end of round $t$, i.e., the $\sigma$-algebra generated by all proposal choices in rounds $1,\dots,t$. Note that $\mathcal H_{t-1}$ determines $U_{t-1}$ and the residual graph $G_{t-1}$. The
probability that a vertex is matched in the next round has a simple closed form:

For $t\geq1$ and  $u\in U_{t-1}$ with $\deg _{t-1} (u)\geq 1$, we have
\[
\Pr(I_t(u)=1 \mid \mathcal H_{t-1})= \sum_{w \in N_{t-1}(u) } \frac{1}{\deg_{t-1}(u)\deg_{t-1}(w)}
\]
Indeed, given $\calH_{t-1}$ the vertices of $U_{t-1}$ propose independently and uniformly to their
current neighbours; $u$ is matched if and only if $\pi_t(u)=w$ and $\pi_t(w)=u$ for some
$w\in N_{t-1}(u)$, and these events are disjoint for distinct $w$.
\smallskip

Finally,  $X \sim \mathrm{Bin}(n,p)$ denotes a binomial random variable with parameters $n$ and $p$, i.e.,
\[
    \Pr(X = k) = \binom{n}{k}p^{k}(1-p)^{n-k} \text{ for } k = 0,1,2,\dots, n.
\]

%% file: warmup.tex

\newpage
\section{Warmup: \bstab~Holds for Few Rounds}\label{sec:warmup}

In this section we show, using only coarse martingale concentration (Azuma--Hoeffding
inequality), 
that the $\beta$-stability invariant holds for the first
$\Theta(\log^{\epsilon/2}n)$ rounds. This serves as a warmup that captures the key difficulty, namely
controlling how many neighbours a vertex loses in each round, and shows where such coarse concentration
stops being enough.

\smallskip

We begin by studying the first round of the process on
a random graph $G \sim \gnp$. This analysis provides local control over the evolution of a vertex neighbourhood within one round.

\medskip
\subsection{Expected Loss of Degree is Bounded}\label{sec:expected-loss}

\begin{lemma}[Expected degree loss in the first round]
\label{lem:exp-matched-neighbours}
 Fix $\epsilon\geq 0$.   Let $G \sim \gnp$ with $p = \frac{c\log^{1+\epsilon} n}{n}$. Consider any vertex $v \in U_0 = V(G)$.
    Then
    \[
        \E[X_1(v)] \le \frac{\Delta_0}{\delta_0}.
    \]
    In particular, with probability $1-o(1)$,
    \[
        \E[X_1(v)] \le O(1).
    \]
\end{lemma}
\begin{proof}

Recall from \Cref{def:xvt}, we have for every $v\in U_0$,
\[
X_1(v) = \sum_{u\in N_0(v)} I_1(u).
\]
Hence
    \begin{align}
        \label{eq:expecx1}
        \E[X_1(v)]
        =
        \sum_{u\in N_0(v)} \Pr(u \in M_1).
    \end{align}
    Fix a neighbour $u\in N_0(v)$. Vertex $u$ is matched in round~1 if and only if there exists a vertex $w\in N_0(u)$ such that $\pi(u)=w$ and $\pi(w)=u$. 
\smallskip

Summing  over all $w\in N_0(u)$ we obtain,
    \begin{align*}
        \Pr(u \in M_1) &= \sum_{w\in N_0(u)}\Pr(\pi(u)=w,\ \pi(w)=u)  \\ 
        & = \sum_{w\in N_0(u)} \frac{1}{\deg_0(u)\deg_0(w)}\\
        &= \frac{1}{\deg_0(u)} \sum_{w\in N_0(u)} \frac{1}{\deg_0(w)}\\
        &\leq \frac{1}{\deg_0(u)} \cdot \frac{\deg_0(u)}{\delta_0} = \frac{1}{\delta_0}.
    \end{align*}
    Substituting into \cref{eq:expecx1} gives
    \begin{align}
        \label{eq:Ex1}
        \E[X_1(v)]
        = \sum_{u\in N_0(v)} \Pr(u \in M_1)
        \leq
        \sum_{u\in N_0(v)} \frac{1}{\delta_0}
        =
        \deg_0(v)\cdot \frac{1}{\delta_0}
        \leq
        \frac{\Delta_0}{\delta_0}.
    \end{align}
    We now condition on the degree-regularity event $\ereg$ from \Cref{fact:deg-conc}, which holds with probability at least $1-o(1)$. Under $\ereg$, every vertex $v \in U_0$ satisfies
    \[
        \delta_0 \geq  \frac{c}{2}\cdot d_0,
        \qquad
        \Delta_0 \leq 2c\cdot d_0.
    \]

    \Cref{eq:Ex1} then implies that
    \[
    \E[X_1(v)]
    \leq
    \frac{\Delta_0}{\delta_0} = \frac{2c\cdot d_0}{(c/2) \cdot d_0} = 4
    \]
    is true with probability at least $1-o(1)$.
\end{proof}

\medskip

\subsection{Degree-Loss in the First Round is Concentrated}\label{sec:no-isolation-first-round}

We now prove a concentration bound for the degree loss in the first round. Our bound will show that with high probability no vertex loses more than half its neighbours.
Since the argument uses martingales, we begin by describing the proposal exposure process.
\smallskip

Fix a vertex $v \in U_0$, and let $d:=\deg_0(v)$ for the purposes of this section. To analyse the number of neighbours $v$ loses in the first round, we separate the randomness according to whether it originates
from vertices inside or outside the neighbourhood of $v$. Since proposal choices are made \emph{independently} by each vertex, we may first
expose the choices made by vertices outside $N_0(v)\cup\{v\}$ and condition
on them.
\smallskip

Let $\mathcal R_1^{\mathrm{out}}(v)$ denote the proposal choices of all vertices outside $N_0(v)\cup\{v\}$, i.e.,
\begin{align}
\mathcal R_1^{\mathrm{out}}(v)
:=
\{\pi(x): x \in U_0 \setminus (N_0(v)\cup\{v\})\}.\notag
\end{align}

\smallskip

Fix an arbitrary realization $r_1^{\mathrm{out}} \leftarrow \mathcal R_1^{\mathrm{out}}$ and condition on it throughout, i.e., we work in the probability space
\begin{align}
    \Pr(\,\cdot \mid \mathcal R_1^{\mathrm{out}}(v) = r_1^{\mathrm{out}}).\notag
\end{align}
Let $X_1^{(r_1^{\mathrm{out}})}(v)$ denote the number of neighbours of $v$ matched in round~1 under this conditioning, i.e.,
\begin{align}
    X_1^{(r_1^{\mathrm{out}})}(v) := X_1(v) \,\big|\, \mathcal R_1^{\mathrm{out}}(v) = r_1^{\mathrm{out}} \notag
\end{align}

Under this conditioning, the proposal choices outside $N_0(v)\cup\{v\}$ are fixed. It therefore suffices to expose the proposal choices of vertices in $N_0(v)=\{u_1,\dots,u_d\}$. For each $i\in[d]$, let $\mathcal N_i$
denote the (fixed) set of neighbours from which $u_i$ chooses.
\smallskip

Let $\Omega := \mathcal N_1 \times \cdots \times \mathcal N_d$ denote the
sample space of proposal choices of the neighbours of $v$, where an
outcome $\omega\in\Omega$ is a tuple
$\omega=(\pi(u_1),\dots,\pi(u_d))$.
\smallskip

We reveal these proposals sequentially and define the corresponding filtration by
\begin{align}
    \fil{1}{k}{v} :=
    \begin{cases}
        \{\emptyset,\Omega\}, & k=0,\\[2mm]
        \sigma\bigl(\pi(u_1),\dots,\pi(u_k)\bigr), & k\in[d].
    \end{cases}\notag
\end{align}
Clearly, $\fil{1}{s}{v}\subseteq \fil{1}{k}{v}$ for all $s\leq k$, and thus $\filseq{1}{k}{v}{d}{}$ is a filtration of the conditional probability space induced by the conditioning $\mathcal R_1^{\mathrm{out}}(v)=r_1^{\mathrm{out}}$.

\smallskip

We consider the Doob martingale $\Zmartseq{1}{k}{v}{d}{}$ with respect to the filtration $\filseq{1}{k}{v}{d}{}$, i.e.,
\begin{align}
    \label{def:martingale}
    \Zmart{1}{k}{v} := \mathbb E\!\left[X_1^{(r_1^{\mathrm{out}})}(v) \mid \fil{1}{k}{v}\right] \qquad \text{for}
    ~~ k\in [d].
\end{align}
We refer to $\Zmartseq{1}{k}{v}{d}{}$ as the \emph{neighbourhood-exposure martingale}.

Note that 
\begin{align}
\Zmart{1}{0}{v} = \mathbb E\!\left[X_1^{(r_1^{\mathrm{out}})}(v)\right]\qquad  \text{and}
\qquad
\Zmart{1}{d}{v} = X_1^{(r_1^{\mathrm{out}})}(v).\notag
\end{align}
We now show that $\Zmartseq{1}{k}{v}{d}{}$ satisfies the bounded difference property.
\begin{lemma}[Bounded differences for the exposure martingale]
\label{lem:bounded-difference-round1}
For the martingale $\Zmartseq{1}{k}{v}{d}{}$ we have
\[
|\Zmart{1}{k}{v} - \Zmart{1}{k-1}{v}| \le 2
\quad \text{for all } k \in [d].
\]
\end{lemma}

\begin{proof}
Fix $k \in [d]$ and condition on the $\sigma$-algebra $\fil{1}{k-1}{v}$.
Under this conditioning, we note that the proposal choices
$\pi(u_1),\dots,\pi(u_{k-1})$ are fixed, while $\pi(u_k)$ and the
choices of $u_{k+1},\dots,u_d$ remain random.
\smallskip

Exposing the proposal $\pi(u_k)$ can only affect edges incident to $u_k$.
Consequently, changing the value of $\pi(u_k)$ can create or destroy
at most one matched edge, say $e$, in the matching produced in round $1$.
\smallskip

If the edge $e$ joins $u_k$ to another neighbour $u_j \in N_0(v)$,
then the matched status of both $u_k$ and $u_j$ may change.
Hence the value of $X_1^{(r_1^{\mathrm{out}})}(v)$
can change by at most $2$.
\smallskip

Therefore, we have
\[
\abs{\Zmart{1}{k}{v} - \Zmart{1}{k-1}{v}}
=
\abs{\mathbb E\!\left[X_1^{(r_1^{\mathrm{out}})}(v) \; \middle| \; \fil{1}{k}{v}\right]
-
\mathbb E\!\left[X_1^{(r_1^{\mathrm{out}})}(v) \; \middle| \; \fil{1}{k-1}{v}\right]}
\le 2.
\]
\end{proof}

The bounded difference property allows us to apply Azuma--Hoeffding inequality,
yielding concentration of the number of neighbours of $v$ matched in
the first round.

\begin{lemma}[Degree-stability in the first round]
    \label{lem:round1-degree-stability}
   Fix $\epsilon\geq 0$. Let $G\sim \gnp$ with $p= \frac{c\log^{1+\epsilon} n}{n}$ for a sufficiently large 
    constant $c$.
    
    Then with probability at least $1-o(1)$, the following holds simultaneously for all vertices $v\in V(G)$:
    \[
    \deg_1(v) \geq \tfrac12 \deg_0(v).
    \]
\end{lemma}
\begin{proof}
Fix a vertex $v\in V(G)$ and write $d=\deg_0(v)$.
Let $X_1^{(r_1^{\mathrm{out}})}(v)$ denote the random variable $X_1(v)$ conditioned on an arbitrary realization
$r_1^{\mathrm{out}} \leftarrow \mathcal R_1^{\mathrm{out}}(v)$ of the proposal choices made by
vertices outside $N_0(v)\cup v$. We consider the martingale defined in \cref{def:martingale}.
\smallskip

By \Cref{lem:bounded-difference-round1}, we have that 
 $\Zmartseq{1}{k}{v}{d}{}$ satisfies the bounded difference property.

\smallskip

Therefore, by the Azuma--Hoeffding inequality (\Cref{fact:azuma}), we obtain:

For any $\lambda>0$,
\begin{align*}
    \Pr\!\left(
    \abs{\Zmart{1}{d}{v} - \Zmart{1}{0}{v}} \geq \lambda
    \;\middle|\; \mathcal R_1^{\mathrm{out}}(v)=r_1^{\mathrm{out}}
    \right)
    \leq
    \exp\!\left(-\frac{\lambda^2}{8d}\right)
\end{align*}
Equivalently,
\begin{align}
    \Pr\!\left(
    X_1^{(r_1^{\mathrm{out}})}(v) - \E\left[X_1^{(r_1^{\mathrm{out}})}(v)\right] \geq \lambda
    \;\middle|\; \mathcal R_1^{\mathrm{out}}(v)=r_1^{\mathrm{out}}
    \right)
    \leq
    \exp\!\left(-\frac{\lambda^2}{8d}\right).
    \label{eq:conditional}
\end{align}

We note that the bound in \cref{eq:conditional} holds uniformly for every realization $r_1^{\mathrm{out}}$. This yields,
\begin{align}
    \label{eq:round1-ineq}
    \Pr\!\left(
    X_1(v) - \E[X_1(v)] \ge \lambda
    \right)
    \leq
    \exp\!\left(-\frac{\lambda^2}{8d}\right).
\end{align}

Conditioning on the degree-regularity event $\ereg$ from \Cref{fact:deg-conc}, which holds with probability $1-o(1)$, we have $d=\deg_0(v)=\Theta(d_0)$ and $\E[X_1(v)]\le 4$. Therefore, choosing $\lambda = \frac{d}{2} - \E[X_1(v)]$ yields $\lambda = \Theta(d_0)$.

Substituting this choice of $\lambda$ into \cref{eq:round1-ineq} yields
\[
\Pr\!\left(X_1(v) \ge \tfrac{d}{2} \mid \ereg \right)
\le n^{-10}.
\]
Applying a union bound over all vertices $v \in V(G)$ gives
\[
\Pr\!\left(\exists v \in V(G): X_1(v) \ge \tfrac12 \deg_0(v)
\;\middle|\; \ereg \right)
\le n^{-9}.
\]
Since $\Pr(\ereg) \ge 1-o(1)$, it follows that with probability $1-o(1)$,
for every vertex $v$ we have
\[
X_1(v) \le \tfrac12 \deg_0(v).
\]
Thus we have,
\[
\deg_1(v)=\deg_0(v)-X_1(v) \ge \tfrac12 \deg_0(v),
\]
with probability at least $1-o(1)$, which proves the lemma.
\end{proof}

\medskip
\subsection{Extending the First-Round Analysis and Iterating Degree-Loss}\label{sec:first-to-arbitrary}

We now extend the first-round analysis to an arbitrary round $t$ and use it to show \bstab~for a few rounds.

\smallskip

Condition on the history $\calH_{t-1}$ and suppose that \bstab~holds. Under this conditioning, the proposal choices made in round $t$ remain independent and uniform over current neighbourhoods.
\smallskip

Fix a vertex $v\in U_{t-1}$ and let $d := \deg_{t-1}(v)$. We expose the proposal choices of the neighbours of $v$ sequentially, inducing a filtration and the corresponding Doob martingale for $X_t(v)$. The same argument therefore yields the following round-$t$ analogues of \Cref{lem:bounded-difference-round1,lem:exp-matched-neighbours}, which we then iterate to control degree evolution.
\begin{lemma}
    \label{lem:exp-matched-neighbours-t}
    Fix $\eps \geq 0$. Let $G \sim \gnp$ with $p = \frac{c\log^{1+\epsilon} n}{n}$. Assume that \bstab~holds up to round $t$. 
    
    Then, for every vertex $v \in U_{t-1}$, 
    \begin{align*}
        \mathbb E[X_t(v) \mid \mathcal{H}_{t-1}] \le \frac{\Delta_{t-1}}{\delta_{t-1}}.
    \end{align*}
\end{lemma}
\begin{lemma}[Concentration of degree-loss in round $t$]
\label{lem:roundt-degree-stability}
Fix $\eps \geq 0$. Let $G \sim \mathds{G}(n, p)$ with $p=\frac{c \log^{1+\epsilon} n} {n}$. Suppose the degree regularity event $\ereg$ holds and there exists an absolute constant $\kappa$ such that:
\[
\frac{\Delta_{t-1}}{\delta_{t-1}} \le \kappa.
\]
Fix any $\lambda > 0$. Condition on $\calH_{t-1}$ and suppose \bstab~holds up to round $t$.
Then for every vertex $v \in U_{t-1}$, we have
\[
\Pr\big( X_t(v) \ge \kappa + \lambda \,\big|\, \calH_{t-1} \big)
\le \exp\!\Big(-\frac{\lambda^2}{8 \deg_{t-1}(v)}\Big).
\]
\end{lemma}

We now iterate these bounds to control the cumulative degree loss over multiple rounds. In particular, we show that with high probability the degree of every
vertex decreases by at most a fixed additive amount in each round.

\begin{lemma}[Additive loss of degree over multiple rounds]
\label{lem:additive-loss-t-round}
Fix $\eps \geq 0$. Let $G \sim \mathbb{G}(n,p)$ with $p = \frac{c \log^{1+\eps} n}{n}$ and let $R := 8 + \sqrt{160 c \log^{2+\epsilon} n}$. Suppose that $TR < \frac{c}{4}\log ^{1+\epsilon} n$.

Then for every round $t \le T$ and every vertex $v \in U_t$, we have 
\[
\deg_t(v) \ge \deg_0(v) - tR
\]
with probability at least $1 - (T+1)\cdot n^{-9}$.
\end{lemma}

Consequently, the minimum degree of the residual graph remains large
for a number of rounds, implying that the \bstab~invariant
holds during this period.

\begin{corollary}[\bstab~for $\Theta(\log^{\epsilon/2} n)$ rounds]
\label{cor:weak-iterative-bound}
Fix $\epsilon \geq 0$. Let $G \sim \mathbb{G}(n,p)$ with $p = \frac{c\log^{1+\epsilon} n}{n}$.
Then \bstab~holds up to
\[
T = \Theta(\log^{\epsilon/2} n)
\]
rounds with probability at least $1-o(1)$.
\end{corollary}

The proofs of \Cref{lem:additive-loss-t-round} and \Cref{cor:weak-iterative-bound} are given in \Cref{app:iterating-deg-pos}.

\begin{remark}
\label{rem:weak-bound-limit}
While \Cref{lem:additive-loss-t-round} and \Cref{cor:weak-iterative-bound} hold for any $\epsilon \ge 0$, they are inherently limited by the coarse concentration of the Azuma--Hoeffding inequality. Specifically, because this approach allows for deviations of order $\sqrt{d_0 \log n}$ in each round, the iterative argument only yields \bstab~for $T = O(\sqrt{d_0 / \log n})$ rounds. In the case where $\epsilon = 0$, this simply ensures that no vertex becomes isolated for a constant number of rounds with high probability.
\end{remark}

%% file: variance-bound.tex
\newpage
\section{Strong Local Concentration via Structural Sparsity of $\gnp$}\label{sec:strong-conc}

In this section we exploit structural sparsity of \erenyi~random graphs to obtain strong concentration bounds for the per-round degree loss.
\smallskip

In contrast to the Azuma--Hoeffding bound obtained earlier, which yields control over deviations of order $\sqrt{d_0\log n}$, we obtain that deviations of order $O(\log n)$  occur with exponentially small tails. This is done by exploiting the structural sparsity of $\gnp$, which ensures that the variance of the degree loss in each round is $O(1)$. Combined with a martingale argument, this yields a Freedman-type concentration bound. This stronger control of per-round degree loss serves as the local ingredient for controlling the evolution of vertex degrees across many rounds.
\smallskip

To achieve this, we must control correlations between the matching events of different
neighbours of a vertex. Two neighbours $u$ and $w$ may either be adjacent or share a common neighbour $z$, introducing correlations between the indicators $I_t(u)$ and $I_t(w)$. We capture these interactions using a dependency graph defined on the neighbours of a vertex.

\begin{definition}[Dependency Graph for In-Round Indicators]
\label{def:dependency_graph}
Fix a round $t$ and a vertex $v \in U_{t-1}$.
The \emph{dependency graph} $\dtv$ is the graph with vertex set
$N_{t-1}(v)$ in which two distinct vertices $u,w \in N_{t-1}(v)$ are
adjacent if either
\begin{itemize}
    \item $u \sim w$, or
    \item there exists a vertex $z \notin \{u,w,v\}$ such that
    $uz, wz \in E(G_{t-1})$.
\end{itemize}
\end{definition}

Bounding the maximum degree of $\dtv$
will be the key structural input for our variance-based concentration arguments. 

We now define the structural event that captures this bounded
dependency property.

\begin{definition}[Structural Event ($\mathcal E_{\mathrm{struct}}$)]
\label{def:structural_event}
  Let $\mathcal E_{\mathrm{struct}}$ denote the event that
  for every vertex $v$ in the graph $G$ and every neighbour
  $u \in N_0(v)$, there is at most one vertex
  $w \in N_0(v)\setminus\{u\}$ such that $u$ and $w$ are adjacent or have a common neighbour other than $v$, i.e.,
\[
\mathcal E_{\mathrm{struct}}
=
\left\{
\forall v \in V(G),\ \forall u \in N_0(v):
\left|
\left\{
w \in N_0(v)\setminus\{u\}
:
\left( u \sim w \ \text{or}\ N_0(u)\cap N_0(w)\setminus\{v\} \neq \emptyset \right)\right\}
\right|
\le 1
\right\}.
\]
\end{definition}

\medskip

We remark that the role of the random graph model $\gnp$ in our analysis
is primarily to guarantee certain structural properties with high probability.
In particular, the key ingredient is the bounded dependency structure
captured by the event $\mathcal E_{\mathrm{struct}}$, which ensures that
the dependency graph among neighbours of a vertex has maximum degree at most one.
This condition can be viewed as enforcing a form of local sparsity,
in which interactions between neighbours are limited to isolated pairs.
\smallskip

Conditioned on such structural properties, the analysis of the matching process
relies only on the randomness of the algorithm itself (i.e., the proposal choices),
rather than on additional randomness from the graph. The $\gnp$ model thus serves
as a convenient framework in which these properties hold with high probability.
\smallskip

It would be interesting to understand to what extent our analysis extends to
broader classes of graphs satisfying similar structural conditions.

\medskip
\subsection{Local Dependency Graph is a Matching}\label{sec:dependancy-graph}
We now show that the structural event $\estr$ holds with high probability in $\gnp$. As a consequence, w.h.p., for every vertex $v$ and round $t$, the corresponding dependency graph $\dtv$ has maximum degree at most $1$, and is therefore a matching.
\begin{lemma}[Bounded Local Dependency]
\label{lem:local_dependency}
Fix $\epsilon \geq 0$. Let $G \sim \gnp$ with
$p=\frac{c\log^{1+\epsilon} n}{n}$.
Then
\[
\Pr(\mathcal E_{\mathrm{struct}}) = 1-o(1).
\]
\end{lemma}
\begin{proof}
We condition on $\mathcal E_{\mathrm{reg}}$,
which holds with probability $1-o(1)$.

Fix a vertex $v$ and consider $u \in N_0(v)$. Let $d = \deg_0(v)$.

We define

\[
Y_0(u) :=
\left|
\left\{
w \in N_0(v)\setminus\{u\}
:
\left( u \sim w \ \text{or}\ N_0(u)\cap N_0(w)\setminus\{v\} \neq \emptyset \right)
\right\}
\right|.
\]
For each $w \in N_0(v)\setminus\{u\}$ define
\[
\diamd{0}{u}{w}
=
\mathds{1}\{
u \sim w \ \text{or}\ \exists z \notin \{u,w,v\}
\text{ such that } uz, wz \in E(G)
\}.
\]
Thus, 
\[
Y_0(u) = \sum_{w} \diamd{0}{u}{w}.
\]
We note that the lemma is equivalent to showing
\[
\Pr(\exists v,u : Y_0(u) \ge 2) = o(1).
\]
The remainder of the proof bounds this probability.

Fix $u,w$ and analyse $\diamd{0}{u}{w}$. We have
\begin{align}
\Pr(\diamd{0}{u}{w}=1)
\le
\Pr(u \sim w) + \Pr(\exists z \notin \{u,w,v\} : uz, wz \in E(G)).\label{eqn:prob-split}
\end{align}
The first term in \cref{eqn:prob-split} is $p$. For the second term, the number of common neighbours
of $u$ and $w$ outside $\{v\}$ is distributed as $\mathrm{Bin}(n-3,p^2)$. Therefore,
\[
\Pr(\exists z : uz, wz \in E)
\le
(n-3)p^2.
\]
Thus we have,
\[
\Pr(\diamd{0}{u}{w}=1)
\le
p + (n-3)p^2
=
O\left(\frac{\log^{2+2\epsilon} n}{n}\right).
\]

This implies,
\[
\E[Y_0(u)]
\le
O\left(d \cdot  \frac{\log^{2+2\epsilon} n}{n}\right)
=
o(1),
\]
where we used $ d_0 = \Theta(\log^{1+\epsilon} n)$.

If $Y_0(u) \ge 2$, then there exist distinct $w,w'$
such that $\diamd{0}{u}{w}=\diamd{0}{u}{w'}=1$.
By union bound we have,
\begin{align}
\Pr(Y_0(u) \ge 2)
\le
\sum_{w\ne w'} \Pr(\diamd{0}{u}{w}=1 \text{ and } \diamd{0}{u}{w'}=1).\label{eq:union1}
\end{align}
For a fixed pair of distinct neighbours $w, w' \in N_0(v) \setminus \{u\}$, 
the joint event $(\diamd{0}{u}{w}=1$ and $\diamd{0}{u}{w'}=1)$
occurs if $u$ interacts with both $w$ and $w'$ either via adjacency
or via shared neighbours.

We bound $\Pr(\diamd{0}{u}{w}=1 \text{ and } \diamd{0}{u}{w'}=1)$
by considering the following cases.

\paragraph{Case A (Adjacency).}

Both interactions occur via direct edges, i.e.,
$u \sim w$ and $u \sim w'$.

This requires the existence of both the edges
$uw$ and $uw'$. Since each such edge appears independently with probability $p$, this event occurs with probability at most $O(p^2)$.

\paragraph{Case B (Mixed Interaction).}

One interaction occurs via adjacency and the other via a shared neighbour.
\smallskip

Without loss of generality, suppose $u \sim w$ and there exists a vertex
$z \notin \{u,w',v\}$ such that $uz, w'z \in E$.

\smallskip

This requires the existence of three edges: $uw, uz,$ and $w'z$.
As there are at most $n$ choices for $z$, we have
\[
\Pr(\exists z : uw, uz, w'z \in E)
\le
n p^3.
\]
The symmetric case where $u \sim w'$ and $u$ shares a neighbour with $w$
is handled identically. Thus, by a union bound,
\[
\Pr(\text{Case B})
\le
2 n p^3.
\]
\paragraph{Case C (Shared Witness).}

A single vertex $z$ is adjacent to all three vertices $u, w,$ and $w'$.

This requires the existence of three distinct edges $uz, wz,$ and $w'z$.
As there are at most $n$ choices for $z$, we have
\[
\Pr \left (\exists z : uz, wz, w'z \in E \right )\leq np^3.
\]
\paragraph{Case D (Distinct Witnesses).}

There exist distinct vertices $z, z'$ such that
$z \in N_0(u) \cap N_0(w)$ and
$z' \in N_0(u) \cap N_0(w')$.
\smallskip

This requires the existence of four distinct edges
$uz, wz, uz',$ and $w'z'$.
As there are at most $n^2$ choices for $(z, z')$, we have
\[
\Pr \left (\exists z, z' : uz, wz, uz', w'z' \in E \right )\leq n^2p^4.
\]

Putting together the probabilities derived in the above cases we have that
\begin{align}
\Pr(\diamd{0}{u}{w}=1 \text{ and } \diamd{0}{u}{w'}=1)
=
O(p^2+3np^3 + n^2 p^4)
=
O\left(\frac{\log^{4+4\epsilon} n}{n^2}\right).\label{eq:union2}
\end{align}
Using \cref{eq:union2} and the fact that there are at most $O(d^2)$ unordered pairs $(w,w')$, the bound in \cref{eq:union1} simplifies to
\[
\Pr(Y_0(u) \ge 2)
=
O\left(
\frac{\log^{6+6\epsilon} n}{n^2}
\right).
\]
As the number of ordered pairs $(v,u)$
is $O(n\log^{1+\epsilon} n)$,
so a union bound gives
\[
\Pr(\exists v,u : Y_0(u) \ge 2)
=
O\left(
\frac{\log^{7+7\epsilon} n}{n}
\right)
=
o(1),
\]
thus proving the lemma.
\end{proof}

Since $G_t$ is always a subgraph of $G$, no new neighbours can appear during the process. Thus, if $\mathcal E_{\mathrm{struct}}$ holds for the graph $G$,
then for every round $t$ and every vertex $v$,
the dependency graph $\dtv$
has maximum degree at most $1$. We thus have the following lemma.

\begin{lemma}[Dynamic Bounded Dependency]
\label{lem:dynamic-dependency}
Assume that the structural event
$\mathcal E_{\mathrm{struct}}$ holds.
Then for every round $t$ and every vertex $v$,
the dependency graph $\dtv$
has maximum degree at most $1$, and is therefore a matching.
\end{lemma}

\medskip
\subsection{Variance of Degree-Loss is Bounded}\label{sec:bounded-variance}

We now use the bounded dependency structure to control correlations between matching events of neighbours of a vertex $v$, thereby bounding the variance of the degree loss.

\smallskip

A priori, since $v$ has $\Theta(d_0)$ neighbours, the variance of $X_t(v)=\sum_{u\in N_{t-1}(v)} I_t(u)$ could be as large as $\Theta(d_0)$. However, the dependencies among the indicators $I_t(u)$ are tightly constrained. Competition among neighbours to match with $v$ is self-limiting, as $v$ can be matched to at most one vertex, while all remaining dependencies are sparse and captured by the dependency graph $\dtv$, which has maximum degree at most one. Together, these properties reduce the conditional variance of $X_t(v)$ to $O(1)$.
\smallskip
\begin{lemma}[Conditional Variance of Degree Loss]
\label{lem:variance_loss_global}
Fix $\epsilon\geq 0$. Let $G \sim \mathbb{G}(n,p)$ with $p = \frac{c\log^{1+\epsilon} n}{n}$.
Assume that the events $\ereg$ and $\estr$ hold.
Fix a vertex $v$ and a round $t$. Condition on the history
$\calH_{t-1}$ and assume that the $\beta$-stability invariant holds up to round $t$.

Then,
\[
\mathrm{Var}(X_t(v) \mid \calH_{t-1}) = O(1).
\]
\end{lemma}

\begin{proof}
By assumption $\estr$ holds. We condition on the history
$\calH_{t-1}$. Under this conditioning the residual graph $G_{t-1}$
and the neighbourhood $N_{t-1}(v)$ are fixed, while the proposal
choices in round $t$ remain random. Suppose that the $\beta$-stability invariant holds up to round $t$, i.e., we have
\begin{equation}
\delta_{t-1} \ge \beta d_0.\label{eq:degree-bounds}
\end{equation}

Recall that
\begin{align}
X_t(v)=\sum_{u\in N_{t-1}(v)} I_t(u).\label{eq:xtv-nbd}
\end{align}

For each $u\in N_{t-1}(v)$ we define the indicators
\begin{align}
\mat{t}{v}{u} := \mathds{1}\{u \text{ is matched with } v \text{ in round } t\},\notag \\
\nmat{t}{v}{u} := \mathds{1}\{u \text{ is matched with some } z \neq v \text{ in round } t\}.\notag
\end{align}

Then 
\begin{align}I_t(u)=\mat{t}{v}{u}+\nmat{t}{v}{u} \notag \\
\mat{t}{v}{u}\cdot \nmat{t}{v}{u}=0.\notag 
\end{align}
\Cref{eq:xtv-nbd} thus simplifies to
\[
X_t(v) = \sum_{u\in N_{t-1}(v)} \mat{t}{v}{u} + \sum_{u\in N_{t-1}(v)} \nmat{t}{v}{u} .
\]
We define
\begin{align}
X_t^{(v)} := \sum_{u\in N_{t-1}(v)} \mat{t}{v}{u}\\
\qquad
X_t^{(\neg v)} := \sum_{u\in N_{t-1}(v)} \nmat{t}{v}{u} .
\end{align}
Hence
\begin{equation}
X_t(v)=X_t^{(v)}+X_t^{(\neg v)}.\label{eq:X-decomposition}
\end{equation}

We analyse the variance of the two components $X_t^{(v)}$ and $X_t^{(\neg v)}$ separately.

\paragraph{Bounding the variance of $X_t^{(v)}$}\hspace{1mm}

\smallskip

For distinct $u,w \in N_{t-1}(v)$ the events $\{\mat{t}{v}{u}=1\}$ and $\{\mat{t}{v}{w}=1\}$ are
mutually exclusive, since both require $\pi_t(v)$ to take different values.
Hence $X_t^{(v)} \in \{0,1\}$ and thus
\begin{align}
\mathrm{Var}\left(X_t^{(v)}\right) \le 1.\label{eq:var-sum-pos}
\end{align}

\paragraph{Bounding the variance of $X_t^{(\neg v)}$}\hspace{1mm}
\smallskip

We now analyse $X_t^{(\neg v)}=\sum_u \nmat{t}{v}{u}$. 

\smallskip

Expanding the variance gives
\[
\mathrm{Var}\left(X_t^{(\neg v)}\right)
=
\sum_{u\in N_{t-1}(v)} \mathrm{Var}(\nmat{t}{v}{u})
+
\sum_{\substack{u,w \in N_{t-1}(v)\\ u\neq w}}\mathrm{Cov}(\nmat{t}{v}{u},\nmat{t}{v}{w}).
\]
Fix a vertex $u  \in N_{t-1}(v)$. As the r.v. $\nmat{t}{v}{u}$ is an indicator, we have
\begin{align}
  \mathrm{Var}(\nmat{t}{v}{u})
&  \le \E[\nmat{t}{v}{u}]\notag\\
 &= \Pr(u \text{ is matched with some } z \neq v \text{ in round } t)\notag\\
  &= \sum_{w \in N_{t-1}(u) \setminus v} \frac{1}{\deg_{t-1}(u)\deg_{t-1}(w)}\notag\\
  &\le \frac{1}{\delta_{t-1}} \cdot \sum_{w \in N_{t-1}(u) \setminus v} \frac{1}{\deg_{t-1}(u)}\notag\\
  &= \frac{1}{\delta_{t-1}} \cdot \frac{\deg_{t-1}(u)-1}{\deg_{t-1}(u)}\notag\\
  &\le \frac{1}{\delta_{t-1}} \notag\\
  &\le
  \frac{1}{\beta d_0}.\label{eq:var-sum1}
\end{align}
Here the last inequality follows from the assumption that \bstab~holds (\cref{eq:degree-bounds}).

\smallskip

As we have at most $\Delta_{t-1}=O(d_0)$ neighbours of $v$, using the bound in \cref{eq:var-sum1} we have
\begin{equation}
\sum_{u\in N_{t-1}(v)} \mathrm{Var}(\nmat{t}{v}{u}) = O(1).\label{eq:variance-sum}
\end{equation}
Dependence between $\nmat{t}{v}{u}$ and $\nmat{t}{v}{w}$ can arise only when $u$ and $w$
interact through a potential matching event, i.e., when either $u \sim w$
or they share a common neighbour $z \neq v$. Equivalently, this occurs
precisely when $\{u,w\}$ is an edge of the dependency graph $\dtv$. Hence the number of unordered pairs
$(u,w)$ for which $\mathrm{Cov}(\nmat{t}{v}{u},\nmat{t}{v}{w}) \neq 0$ is at most the number
of edges of $\dtv$.

\smallskip

By \Cref{lem:dynamic-dependency}, $\dtv$ is a matching, so it has
at most $O(d_0)$ edges and therefore we have at most $O(d_0)$ dependent pairs.
For each such pair, we have
\begin{align}
|\mathrm{Cov}(\nmat{t}{v}{u},\nmat{t}{v}{w})|
&\le
\min\left(\E[\nmat{t}{v}{u}], \E[\nmat{t}{v}{w}]\right)\notag \\
&\le
\frac{1}{\beta d_0},\label{eq:bd0}
\end{align}
where the last inequality (\cref{eq:bd0}) again uses the assumption that \bstab~holds.

Thus,
\begin{align}
\sum_{\substack{u,w \in N_{t-1}(v)\\ u\neq w}}
|\mathrm{Cov}(\nmat{t}{v}{u},\nmat{t}{v}{w})|
=
O(d_0)\cdot O\!\left(\frac{1}{\beta d_0}\right)
=
O(1).\label{eq:covariance-sum}
\end{align}
Combining the bounds in \cref{eq:variance-sum,eq:covariance-sum} yields
\begin{align}
\mathrm{Var}\left(X_t^{(\neg v)}\right) = O(1).\label{eq:var-sum-neg}
\end{align}
\paragraph{Conclusion.}

Writing out the variance of $X_t(v)$ using the decomposition in \cref{eq:X-decomposition} we have
\[
\mathrm{Var}(X_t(v))
=
\mathrm{Var}\left(X_t^{(v)}\right)
+
\mathrm{Var}\left(X_t^{(\neg v)}\right)
+
2\,\mathrm{Cov}\left(X_t^{(v)},X_t^{(\neg v)}\right).
\]
As we have shown in \cref{eq:var-sum-pos,eq:var-sum-neg} that $\mathrm{Var}\left(X_t^{(v)}\right) \le 1$ and 
$\mathrm{Var}\left(X_t^{(\neg v)}\right) = O(1)$ respectively, using the Cauchy--Schwarz inequality
implies
\[
\left| \mathrm{Cov}\left(X_t^{(v)},X_t^{(\neg v)}\right)\right|
\le
\sqrt{\mathrm{Var}\left(X_t^{(v)}\right)\cdot \mathrm{Var}\left(X_t^{(\neg v)}\right)}
= O(1).
\]
Thus, we have shown $\mathrm{Var}(X_t(v)) = O(1)$, which completes the proof.
\end{proof}

\medskip
\subsection{Strong Concentration of Degree-Loss via Variance Control}\label{sec:strong-conc-variance}
We are now ready to establish stronger concentration for the degree loss
in a single round under the $\beta$-stability invariant.
 This provides control of local deviations within a round and serves as the key input for controlling the evolution of degrees over many rounds.

\begin{lemma}[Strong One-Round Concentration of Degree Loss]
\label{lem:bernstein-inround}
Fix $\epsilon\geq 0$. Let $G \sim \mathbb{G}(n,p)$ with $p = \frac{c\log^{1+\epsilon} n}{n}$. Assume that the events
$\ereg$ and $\estr$ hold.
Fix a vertex $v$ and a round $t$. Condition on the history 
$\calH_{t-1}$ and suppose the $\beta$-stability invariant holds up to round $t$.

Then there exists absolute constant $C>0$,
such that for all $\lambda \ge 0$,
\[
\Pr\!\left(
\left|X_t(v) - \E[X_t(v)\mid \calH_{t-1}]\right| 
\ge \lambda
\mid \calH_{t-1}
\right)
\le
2\exp\!\left(
-\frac{\lambda^2}{C(1+\lambda)}
\right).
\]
In particular, there exists a sufficiently large absolute constant $C_1$
such that
\[
\Pr\!\left(
X_t(v) \ge C_1 \log n
\mid \calH_{t-1}
\right)
\le O(n^{-12}).
\]
\end{lemma}
\begin{proof}

Condition on $\calH_{t-1}$ and assume that $\estr$
and the $\beta$-stability invariant hold up to round $t$.

We analyze $X_t(v)$ using a martingale argument that exploits the
bounded dependency structure of the indicators $\{I_t(u) : u \in N_{t-1}(v)\}$.

\smallskip

By \Cref{lem:dynamic-dependency}, the dependency graph $\dtv$ has
maximum degree at most $1$, and hence is a matching. Therefore,
$N_{t-1}(v)$ can be partitioned into disjoint blocks
$B_1,\dots,B_m$ of size at most $2$, where each block contains all
mutually dependent indicators.

For each $i \in [m]$, define
\begin{align}
Y_i := \sum_{u \in B_i} I_t(u).\notag
\end{align}
Thus,
\begin{align}
X_t(v) = \sum_{i=1}^m Y_i,
\label{eq:decomp-X}
\end{align}
with $0 \le Y_i \le 2$.

\smallskip

In what follows, we will condition on the proposal $\pi_t(v)$. Conditioning on $\pi_t(v)$ eliminates the only source of global
dependency among the indicators $\{I_t(u)\}_{u \in N_{t-1}(v)}$,
ensuring that all remaining dependencies are captured by the local
dependency graph $\dtv$.

\smallskip

Define,
\begin{align}
\widetilde{Y}_i := Y_i - \E[Y_i \mid \calH_{t-1}, \pi_t(v)].
\notag
\end{align}
Then,
\begin{align}
|\widetilde{Y}_i| \le 2
\quad  \text{and} \quad 
\E[\widetilde{Y}_i \mid \calH_{t-1}, \pi_t(v)] = 0. \label{eq:Ytilde}
\end{align}
Let
\begin{align}
Z := \E[X_t(v) \mid \calH_{t-1}, \pi_t(v)].\notag
\end{align}
Then,
\begin{align}
\sum_{i=1}^m \widetilde{Y}_i = X_t(v) - Z.
\label{eq:centered-sum}
\end{align}
Since changing $\pi_t(v)$ affects the matching status of at most two
neighbours of $v$, we have
\begin{align}
\big|Z - \E[X_t(v)\mid \calH_{t-1}]\big| \le 2.
\label{eq:mean-shift}
\end{align}Using  triangle inequality, we have for any $\lambda \ge 2$,
\begin{align}
|X_t(v) - \E[X_t(v)\mid \calH_{t-1}]| \ge \lambda
\quad \;\Rightarrow\;\quad 
\left|
\sum_{i=1}^m \widetilde{Y}_i
\right| 
\ge \lambda - 2.
\label{eq:triangle-reduction}
\end{align}
Let
\begin{align}
\sigma^2 :=
\Var\!\left(\sum_{i=1}^m \widetilde{Y}_i \,\middle|\, \calH_{t-1}, \pi_t(v)\right).
\notag
\end{align}
Since $Z$ is determined by $\calH_{t-1}$ and $\pi_t(v)$, $\sigma^2$ is the conditional variance of the
degree loss $X_t(v)$ once the proposal of $v$ is fixed. \Cref{lem:variance_loss_global} bounds the
variance given $\calH_{t-1}$ alone; by the law of total variance, this bounds only the average of
$\sigma^2$ over the values of $\pi_t(v)$. We show below that $\sigma^2=O(1)$ for every value of
$\pi_t(v)$ (see \cref{eq:sigma-bound}); this follows from the bound on the quadratic variation of a
block martingale, which is also the bound used in Freedman's inequality.
Fix $\calH_{t-1}$ and $\pi_t(v)$. Let $\fil{t}{k}{v}[G]$ be the $\sigma$-algebra generated by $\calH_{t-1}$,
$\pi_t(v)$, and all proposal variables
$\{\pi_t(x) : x \in \bigcup_{j=1}^k (B_j \cup N_{t-1}(B_j))\}$.

\smallskip

Each $Y_i$ is determined by the proposal variables of vertices in $B_i$
and their neighbours; hence $\fil{t}{k}{v}[G]$ reveals $Y_1,\dots,Y_k$.

\smallskip

Define the Doob martingale
\begin{align}
\Mart{t}{k}{v}[W] := \E\!\left[ \sum_{i=1}^m \widetilde{Y}_i \,\middle|\, \fil{t}{k}{v}[G] \right],
\qquad k=0,1,\dots,m.
\label{eq:martingale}
\end{align}
Then $\Mart{t}{0}{v}[W] = 0$ and $\Mart{t}{m}{v}[W] = \sum_{i=1}^m \widetilde{Y}_i$.

Let
\begin{align}
D_k := \Mart{t}{k}{v}[W] - \Mart{t}{k-1}{v}[W].
\label{eq:increments}
\end{align}
Since $\fil{t}{k}{v}[G]$ reveals all randomness determining $Y_k$, we have
\[
D_k = Y_k - \E[Y_k \mid \fil{t}{k-1}{v}[G]],
\]
Using \cref{eq:Ytilde}, we have $|\widetilde{Y}_i| \le 2$, and hence
\begin{align}
|D_k| \le 2 \quad \textrm{for all $k$.}
\label{eq:bounded-increment}
\end{align}
Let
\begin{align}
V_m(v) := \sum_{k=1}^m \E[D_k^2 \mid \fil{t}{k-1}{v}[G]].
\label{eq:quad-var}
\end{align}
We show that $V_m(v) = O(1)$ almost surely.

\medskip

Fix $k$ and condition on $\fil{t}{k-1}{v}[G]$. By construction of the blocks as
connected components of the dependency graph $\dtv$, all proposal variables
that influence the indicators $\{I_t(u) : u \in B_k\}$ are contained within
the block $B_k$. In particular, vertices in earlier blocks cannot introduce
new proposal conflicts with vertices in $B_k$.

\smallskip

Moreover, any vertex $w \in N_{t-1}(v)$ whose proposal could affect the
matching outcome of a vertex $u \in B_k$ must lie in the same block $B_k$.
Indeed, if such a vertex $w$ matches a neighbour $x$ of $u$, then
$w \sim x$ and $u \sim x$, so $w$ and $u$ share a common neighbour.
By the definition of the dependency graph $\dtv$, this implies that
$w$ and $u$ are adjacent in $\dtv$, and hence belong to the same block.

\smallskip

Therefore, conditioning on $\fil{t}{k-1}{v}[G]$ reveals no proposal variables that
affect the matching outcomes of vertices in $B_k$. In particular, for every
$u \in B_k$,
\begin{align}
\Pr(I_t(u)=1 \mid \fil{t}{k-1}{v}[G])
=
\Pr(I_t(u)=1 \mid \calH_{t-1}, \pi_t(v)).\notag 
\end{align}
As \bstab~implies $\Pr(I_t(u)=1 \mid \calH_{t-1}, \pi_t(v))=O(1/d_0)$, we have for every $u \in B_k$,
\begin{align}
\Pr(I_t(u)=1 \mid \fil{t}{k-1}{v}[G]) = O(1/d_0),\label{eq:itu-cond}
\end{align}
uniformly over $\fil{t}{k-1}{v}[G]$.

We are now ready to bound $\E[D_k^2 \mid \fil{t}{k-1}{v}[G]]$.
\begin{align}
\E[D_k^2 \mid \fil{t}{k-1}{v}[G]]
&= \E\big[ Y_k^2 \mid \fil{t}{k-1}{v}[G] \big]
- \E[Y_k \mid \fil{t}{k-1}{v}[G]]^2 \le \E\big[ Y_k^2 \mid \fil{t}{k-1}{v}[G] \big].
\label{eq:dk12}
\end{align}
As $0 \le Y_k \le 2$, we have $Y_k^2 \le 2 Y_k$. Using this in  \cref{eq:dk12}  we obtain
\begin{align}
\E[D_k^2 \mid \fil{t}{k-1}{v}[G]]
\le 2 \cdot \E[Y_k \mid \fil{t}{k-1}{v}[G]].
\notag
\end{align}
Since $Y_k = \sum_{u \in B_k} I_t(u)$ with $|B_k| \le 2$, we have
\begin{align}
\E[Y_k \mid \fil{t}{k-1}{v}[G]]
= \sum_{u \in B_k} \Pr(I_t(u)=1 \mid \fil{t}{k-1}{v}[G]).
\notag
\end{align}
Using \cref{eq:itu-cond} we obtain
\begin{align}
\E[D_k^2 \mid \fil{t}{k-1}{v}[G]] = O(1/d_0).
\label{eq:step-variance}
\end{align}
Summing \cref{eq:step-variance} over all $m = O(d_0)$ blocks yields,
\begin{align}
V_m(v) = \sum_{k=1}^m \E[D_k^2 \mid \fil{t}{k-1}{v}[G]] = O(1).
\label{eq:Vm-bound}
\end{align}

In particular, $\sigma^2=O(1)$. Indeed, $\fil{t}{0}{v}[G]$ is generated by $\calH_{t-1}$ and $\pi_t(v)$,
and the martingale starts at $\Mart{t}{0}{v}[W]=0$ and ends at
$\Mart{t}{m}{v}[W]=\sum_{i=1}^m\widetilde Y_i$, so
$\sum_{i=1}^m\widetilde Y_i=\sum_{k=1}^m D_k$. By \cref{eq:Ytilde} this sum has conditional mean zero,
so its variance equals its conditional second moment. Moreover, for $j<k$ the increment $D_j$ is
$\fil{t}{k-1}{v}[G]$-measurable and $\E[D_k\mid\fil{t}{k-1}{v}[G]]=0$, hence
$\E[D_jD_k\mid\fil{t}{0}{v}[G]]=0$ by the tower property. The cross terms therefore vanish, and
\begin{align}
\sigma^2
&=\E\!\left[\Big(\sum_{i=1}^m\widetilde Y_i\Big)^2 \,\middle|\, \calH_{t-1}, \pi_t(v)\right]
=\E\!\left[\Big(\sum_{k=1}^m D_k\Big)^2 \,\middle|\, \calH_{t-1}, \pi_t(v)\right]\notag\\
&=\E\!\left[\E\!\left[\Big(\sum_{k=1}^m D_k\Big)^2 \, \middle|\, \fil{t}{k-1}{v}[G]\right] \,\middle|\, \calH_{t-1}, \pi_t(v)\right]
=\E\!\left[V_m(v) \,\middle|\, \calH_{t-1}, \pi_t(v)\right]
=O(1)
\label{eq:sigma-bound}
\end{align}
by \cref{eq:Vm-bound}.

The bound in \cref{eq:Vm-bound} holds for every realization, so in \Cref{fact:freedman} the event
$\{V_m(v)\le C_V\}$, where $C_V$ is the constant in \cref{eq:Vm-bound}, has probability one, and the
probability of the joint event there is the probability of the deviation alone. Applying
\Cref{fact:freedman} to the martingale $\big(\Mart{t}{k}{v}[W]\big)_{k=0}^{m}$ and to its negative,
with $R=2$ by \cref{eq:bounded-increment} and $\sigma^2=C_V$, we obtain for every $\mu>0$
\[
\Pr\!\left(
\left|\sum_{i=1}^m \widetilde{Y}_i\right| \ge \mu
\;\middle|\;
\calH_{t-1}, \pi_t(v)
\right)
\le
2\exp\!\left(-\frac{\mu^2}{2C_V+\frac{4}{3}\mu}\right).
\]
Let $C:=\max\{8C_V,6\}$. For $\lambda\ge 4$ we take $\mu=\lambda-2$, so that $\lambda/2\le\mu\le\lambda$;
since $C(1+\lambda)\ge 8C_V+\frac{16}{3}\lambda$, this gives
\begin{align}
\Pr\!\left(
\left|
\sum_{i=1}^m \widetilde{Y}_i
\right|
\ge \lambda - 2
\;\middle|\;
\calH_{t-1}, \pi_t(v)
\right)
\le
2\exp\!\left(
-\frac{\lambda^2}{C(1+\lambda)}
\right)
\qquad\text{for all } \lambda\ge 4.
\label{eq:freedman123}
\end{align}
By \cref{eq:triangle-reduction} and \cref{eq:freedman123}, we obtain for all $\lambda\ge4$
\begin{align}
\Pr\!\left(
|X_t(v) - \E[X_t(v)\mid \calH_{t-1}]|
\ge \lambda
\;\middle|\;
\calH_{t-1}, \pi_t(v)
\right)
\le
2\exp\!\left(
-\frac{\lambda^2}{C(1+\lambda)}
\right).\notag
\end{align}
Since the bound holds uniformly over all realizations of $\pi_t(v)$,
we obtain
\begin{align}
\Pr\!\left(
|X_t(v) - \E[X_t(v)\mid \calH_{t-1}]|
\ge \lambda
\;\middle|\;
\calH_{t-1}
\right)
\le
2\exp\!\left(
-\frac{\lambda^2}{C(1+\lambda)}
\right)\notag
\end{align}
for all $\lambda\ge4$. For $0\le\lambda<4$ the right-hand side is at least $2e^{-16/30}>1$, since
$C\ge6$, so the bound holds trivially. This proves the first claim of the lemma.
Finally, setting
\[
\lambda
=
C_1 \log n
-
\E[X_t(v) \mid \calH_{t-1}],
\]
and using $\E[X_t(v)\mid \calH_{t-1}] = O(1)$
(\Cref{lem:exp-matched-neighbours-t}),
we obtain for sufficiently large $C_1$,
\[
\frac{\lambda^2}{C(1+\lambda)} \ge 12 \log n.
\]
This implies,
\[
\Pr\!\left(
X_t(v) \ge C_1 \log n
\;\middle|\;
\calH_{t-1}
\right)
\le O(n^{-12}),
\]
thus proving the lemma.
\end{proof}

%% file: linear-rounds.tex

\newpage

\section{Controlling Degree Evolution via Cumulative-Deviation Martingales}
\label{sec:deg-evolution}

Although the expected degree loss in each round is bounded as seen in the previous section, small random deviations could in principle accumulate over time. We control this accumulation by tracking the difference between the actual degree loss and its conditional expectation using a martingale. Applying Freedman's inequality (\Cref{fact:freedman}), together with the $O(1)$ variance bound and $O(\log n)$ deviation bounds from \Cref{sec:strong-conc}, shows that these cumulative deviations remain small for $T=\Theta(d_0)$ rounds. Consequently, the $\beta$-stability invariant holds throughout these rounds with high probability.

\smallskip

To formalize when the $\beta$-stability invariant may fail, we begin by defining the notion of threshold time for each vertex, which records the first round in which its degree drops below the \bstab~threshold. We then define a global stopping time corresponding to the earliest such violation over all vertices.

\begin{definition}[Threshold time]
\label{def:threshold-time}
    Let $\beta > 0$ be a constant. For each vertex $v \in V(G)$, define the
threshold time
    \begin{align}
        \tau(v)
        :=
        \min \left\{
        t \ge 0 :
        v \in U_t
        \text{ and }
        \deg_{t+1}(v) < \beta d_0
        \right\},
    \end{align}
    and $\tau(v) = \infty$ if no such time exists.
\end{definition}

\begin{definition}[Critical vertex]
\label{def:critical-vertex}
Let $\crtcl$ denote the lexicographically first vertex attaining
\[
\min_{v \in V(G)} \tau(v).
\]
We call $\crtcl$ the \emph{critical vertex}.
\end{definition}

\begin{definition}[Global stopping time]
\label{def:global-stopping-time}
We define the stopping time
\[
\tau := \tau(\crtcl).
\]
\end{definition}

We note that $\tau$ is the last round before the \bstab~invariant is violated for some vertex. In particular, $\tau \le \tau(v)$ for every $v \in V(G)$. Moreover, 
the threshold time $\tau(v)$ is finite only if $v$ remains unmatched until its degree drops below the threshold. If $v$ becomes matched before its degree falls below the threshold, then {$\tau(v)=\infty$}, and the $\beta$-stability invariant is never violated for that vertex.

Consequently, the analysis only needs to control the process up to the stopping time $\tau$, during which the \bstab~invariant holds.

\smallskip

To analyse the concentration of the process, we define a martingale that tracks the deviation between the actual cumulative degree loss of the critical vertex and its conditional expectation.
\begin{definition}[Cumulative-Deviation Martingale]
\label{def:critical-martingale}
For the critical vertex $\crtcl$, we define the stochastic process $\Martseq{t}{\crtcl}{\N}$
by
\begin{align}
\Mart{t}{}{\crtcl}:=
\begin{cases}
0 & t=0, \\
\displaystyle \sum_{k=1}^t
\left(
X_k(\crtcl)-
\E[X_k(\crtcl)\mid \calH_{k-1}]
\right) & t \ge 1 .
\end{cases}
\end{align}
\end{definition}

We now verify that the process defined above forms a martingale with
respect to the natural filtration of the process.

\begin{lemma}
    $\Martseq{t}{\crtcl}{\N}$ is a martingale with respect to the filtration
$\Hseq{t}{\N}$, i.e.,
    \begin{align}
        \E\left[\Mart{t}{}{\crtcl}\mid\calH_{t-1}\right] = \Mart{t-1}{}{\crtcl}\quad  \text {for ~all}
        \quad t \ge 1.
    \end{align}
\end{lemma}
\begin{proof}
By definition,
\[
\Mart{t}{}{\crtcl}=\Mart{t-1}{}{\crtcl}+X_t(\crtcl)-\E\left[X_t(\crtcl)\mid\calH_{t-1}\right].
\]
Taking conditional expectation with respect to $\calH_{t-1}$ gives
\[
\E\left[\Mart{t}{}{\crtcl}\mid\calH_{t-1}\right]
=
\Mart{t-1}{}{\crtcl}
+
\E\left[X_t(\crtcl)\mid\calH_{t-1}\right]
-
\E\left[X_t(\crtcl)\mid\calH_{t-1}\right]
=
\Mart{t-1}{}{\crtcl},
\]
 where the last equality holds as   $\Mart{t-1}{}{\crtcl}$ is fully determined by $\calH_{t-1}$.
\end{proof}

To restrict attention to the regime where the \bstab~invariant holds, we consider the martingale $\Martseq{t}{\crtcl}{\N}$ stopped at 
time $\tau$. This avoids the need for a longer and more involved round-by-round inductive argument and allows us to analyze the cumulative deviations directly. The stopped martingale ensures all bounds relying on \bstab~remain valid till our target round $T$.

\smallskip

\begin{definition}[Stopped Martingale]
\label{def:stopped-martingale}
Let  $t \wedge \tau :=\min\{t, \tau\}$. For the critical vertex $\crtcl$, we define the stochastic process $\tildeMart{t}{\crtcl}$ as
\[
\tildeMart{t}{\crtcl} := \Mart{t \wedge \tau}{}{\crtcl}.
\]
\end{definition}

We note that as $\E\left[\tildeMart{t}{\crtcl} \mid \calH_{t-1}\right] = \tildeMart{t-1}{\crtcl}$, the stopped process $\tildeMartseq{t}{\crtcl}{\N}$ is also a martingale
with respect to the filtration $\Hseq{t}{\N}$.

\smallskip

In order to apply Freedman's inequality to the stopped martingale, we need to
establish two properties: a uniform bound on its increments, and a bound
on its predictable quadratic variation. We prove these in the following
subsections.
\medskip
\subsection{Bounding the Increments of the Cumulative-Deviation Martingale}\label{sec:increments-cumulative}

We first establish a uniform bound on the increments of the stopped martingale. This step lifts the strong per-round concentration bounds from \Cref{lem:bernstein-inround} to control the cumulative-deviation martingale: the local $O(\log n)$ bounds on the degree loss imply that the martingale increments are uniformly bounded with high probability. We formalize this through the following bounded increment event.

\begin{definition}[Bounded Increment Event]
\label{def:martingale-increment}
We define $\mathcal E_{\mathrm{incr}}(T,R)$ to be the event
\[
\mathcal E_{\mathrm{incr}}(T,R)
:=
\left\{
\forall t \le T :
\left|
\tildeMart{t}{\crtcl} - \tildeMart{t-1}{\crtcl}
\right|
\le R
\right\}.
\]
\end{definition}

\begin{lemma}[Uniform Martingale Increment Bound]
\label{lem:martingale-jump}
Assume $\mathcal E_{\mathrm{struct}}$ and $\mathcal E_{\mathrm{reg}}$ hold.
Fix $T = \alpha d_0$ for a constant $\alpha>0$.
Then for sufficiently large constant $C_1>0$ we have,
\[
\Pr(\mathcal E_{\mathrm{incr}}^c(T,C_1 \log n))
= O(n^{-10}).
\]
\end{lemma}

\begin{proof}
We recall the definition of the stopped martingale,
\[
\tildeMart{t}{\crtcl} - \tildeMart{t-1}{\crtcl}
=
\begin{cases}
X_t(\crtcl) - \mathbb E[X_t(\crtcl)\mid \calH_{t-1}] & t \le \tau, \\
0 & t > \tau.
\end{cases}
\]
Thus, for $ t > \tau$ the lemma holds trivially, so we focus on the case
when  $t \leq \tau$. 

Whenever $t \le \tau$, the \bstab~invariant holds and we can apply \Cref{lem:bernstein-inround}. This yields
\[
\Pr\!\left(
\left|
X_t(\crtcl)
-
\E\left[X_t(\crtcl) \mid \mathcal H_{t-1}\right]
\right|
\ge C_1 \log{n}
\mid
\mathcal H_{t-1}
\right)
= O(n^{-12}),
\]
for sufficiently large $C_1$.

\smallskip

Finally, as $\tildeMart{t}{\crtcl} - \tildeMart{t-1}{\crtcl}=X_t(\crtcl)
-
\E\left[X_t(\crtcl) \mid \mathcal H_{t-1}\right] $, taking a union bound over all
$T$ rounds then gives
\[
\Pr(\mathcal E_{\mathrm{incr}}^c(T,C_1 \log n))
\le
T \cdot n^{-12}
=
O(n^{-10}).
\]
This proves the lemma.
\end{proof}
\medskip
\subsection{Bounding the Quadratic Variation}\label{sec:bounding-quadratic-var}

We next define the predictable quadratic variation of the stopped
martingale in our setting, and then prove that it admits a uniform bound. 
\begin{definition}[Predictable Quadratic Variation]
\label{def:predictable-variation}
We define the predictable quadratic variation $V_T(\crtcl)$ of the stopped martingale $\Martseq{t}{\crtcl}{\N}$ as
\[
V_T(\crtcl)
:=
\sum_{t=1}^{T}
\mathrm{Var}\!\left(
\tildeMart{t}{\crtcl} - \tildeMart{t-1}{\crtcl}
\;\middle|\;
\calH_{t-1}
\right).
\]
\end{definition}

We next control the predictable quadratic variation of the stopped martingale, which captures how variance accumulates over many rounds. Using the $O(1)$ conditional variance bounds established above, we show that this quantity grows at most linearly in the number of rounds.
\begin{lemma}[Quadratic Variation Bound]
\label{lem:quadratic-variation}
Suppose that the events $\ereg$ and $\estr$ hold.

Then for every $T$, we have
\[
V_T(\crtcl) \le C' T
\]
for some constant $C' > 0$.
\end{lemma}

\begin{proof}
We recall the definition of the stopped martingale,
\[
\tildeMart{t}{\crtcl} - \tildeMart{t-1}{\crtcl}
=
\begin{cases}
X_t(\crtcl) - \mathbb E[X_t(\crtcl)\mid \calH_{t-1}] & t \le \tau, \\
0 & t > \tau.
\end{cases}
\]
Hence,
\[
\mathrm{Var}\!\left(
\tildeMart{t}{\crtcl} - \tildeMart{t-1}{\crtcl}
\;\middle|\;
\calH_{t-1}
\right)
=
\mathds{1}_{\{t \le \tau\}}
\mathrm{Var}\!\left(
X_t(\crtcl)
\;\middle|\;
\calH_{t-1}
\right).
\]
By definition of stopping time $\tau$, the \bstab~invariant holds for all vertices.
Thus, \Cref{lem:variance_loss_global} applies, and we have
\[
\mathrm{Var}(X_t(\crtcl)\mid \calH_{t-1}) \le C'
\]
whenever $t \leq  \tau$.

We then have
\[
V_T(\crtcl)
=
\sum_{t=1}^{T}
\mathds{1}_{\{t \le \tau\}}
\mathrm{Var}(X_t(\crtcl)\mid \calH_{t-1})
\le
C' \cdot (T \wedge \tau)
\le
C' T,
\]
proving the lemma.
\end{proof}

\medskip
\subsection{Concentration of Cumulative Degree-Loss via Freedman's Inequality}\label{sec:freedman}

We combine the increment and quadratic variation bounds to apply Freedman's inequality to the stopped martingale, yielding a deviation bound for the cumulative degree loss.
\begin{lemma}[Freedman Deviation Bound]
\label{lem:freedman-deviation}
Fix $\epsilon>0$. Let $T = \alpha d_0$ for a constant $\alpha>0$ and let $R=C_1 \log n$ for a large constant $C_1$. Suppose that the events $\ereg$, $\estr$
and $\mathcal E_{\mathrm{incr}}(T,R)$ hold.

Then for the stopped martingale $\tildeMartseq{t}{\crtcl}{\N}$ we have,
\[
\Pr\!\left(
\max_{k \le T}
\left|\tildeMart{k}{\crtcl}\right|
\ge \frac{d_0}{4}
\right)
\le
\exp\!\big(-\Theta(\log^{\epsilon} n)\big).
\]
\end{lemma}
\begin{proof}
As the  event $\mathcal E_{\mathrm{incr}}(T,C_1 \log n)$ holds, we have that the martingale  $\tildeMartseq{t}{\crtcl}{\N}$ satisfies
\begin{align}
\left|
 \tildeMart{t}{\crtcl} - \tildeMart{t-1}{\crtcl}
\right|
\le 
C_1 \log n \label{eq:Mt1}
\quad
\text{for all } t \le T .
\end{align}
Since the events $\mathcal E_{\mathrm{reg}}$ and $\mathcal E_{\mathrm{struct}}$ hold, \Cref{lem:quadratic-variation} implies 
\begin{align}
V_T(\crtcl)
\le
C' \alpha d_0 =: \sigma^2.\label{eq:Vt1}
\end{align}
The bound in \cref{eq:Vt1} holds for every realization of the process, so the event
$\{V_T(\crtcl)\le\sigma^2\}$ in \Cref{fact:freedman} has probability one, and the probability of the
joint event there is the probability of the deviation alone.
Plugging in the bounds from \cref{eq:Mt1} and \cref{eq:Vt1} and $
\Lambda := \frac{d_0}{4}$ into Freedman's inequality
(\Cref{fact:freedman}) we get,
\[
\Pr\!\left(
\max_{k \le T}
\left|\tildeMart{k}{\crtcl}\right|
\ge \frac{d_0}{4}
\right)
\le
2 \exp\!\left(
-\frac{\Lambda^2}
{2(\sigma^2 + R\Lambda/3)}
\right).
\]
As
\[
\Lambda^2 = \frac{d_0^2}{16}\quad \text{and} 
\quad
\sigma^2 + R\Lambda/3
= O(d_0 \log n),
\]
we have
\[
\frac{\Lambda^2}{2(\sigma^2 + R\Lambda/3)}
=
\Theta\!\left(
\frac{d_0}{\log n}
\right).
\]
Finally, since $d_0 = \Theta(\log^{1+\epsilon} n)$, we get
\[
\Pr\!\left(
\max_{k \le T}
\left|\tildeMart{k}{\crtcl}\right|
\ge \frac{d_0}{4}
\right)
\le
\exp\!\big(-\Theta(\log^{\epsilon} n)\big),
\]
 proving the lemma.
\end{proof}
Thus, cumulative deviations remain $O(d_0)$
over $T=\Theta(d_0)$ rounds.

\medskip
\subsection{Putting Everything Together: \bstab~Holds for $\Omega(d_0)$ Rounds}\label{sec:putting-all-together}

We now combine the deviation bound from \Cref{lem:freedman-deviation} with the expectation estimates for
the degree loss derived in \Cref{lem:exp-matched-neighbours-t} to prove that the $\beta$-stability
invariant persists for $T=\Theta(d_0)$ rounds.

\smallskip

The key idea is that while the expected degree loss in each round is small, tight martingale control prevents cumulative deviations from causing vertex degrees to drop below a constant fraction of their initial values.

\smallskip

Intuitively, the strong per-round concentration bounds control the local fluctuations of degree loss, while the cumulative-deviation martingale  ensures that these fluctuations do not accumulate enough to destabilize vertex degrees over time.

\begin{theorem}[Persistence of $\beta$-stability for $\Omega(d_0)$ rounds]
\label{thm:persistence}
Fix any $\epsilon>0$. Let $G \sim \mathbb{G}(n,p)$ with $p = \frac{c \log^{1+\epsilon} n}{n}$ for a sufficiently large
constant $c>0$, and let $d_0 := \log^{1+\epsilon} n$. There exist constants
$\alpha,\beta>0$ such that, for $T = \alpha d_0$, the following holds with probability $1-o(1)$:

For all rounds $t \le T$ and all vertices $v \in U_t$,
\[
\deg_t(v) \ge \beta d_0 .
\]
\end{theorem}

\begin{proof}

If $\tau \geq T$, then by definition of stopping time the
$\beta$-stability invariant holds for all rounds $t \le T$.
Therefore we only need to bound the probability that the stopping time occurs within the first $T$ rounds. In particular, we show that for $T = \alpha d_0$, we have
\[
\Pr(\tau < T) = o(1).
\]

Let $R = C_1 \log n$ for a sufficiently large constant $C_1$.
We condition on the high-probability events
$\mathcal E_{\mathrm{reg}}, \mathcal E_{\mathrm{struct}}$, and
$\mathcal E_{\mathrm{incr}}(T,R)$.
\smallskip

If \bstab~fails before round $T$, vertex $\crtcl$ witnesses the earliest such failure. For the critical vertex $\crtcl$ and any $t-1\leq \tau$ with $\crtcl\in U_t$, we have
\begin{align}
\deg_t(\crtcl)
=
\deg_0(\crtcl)
-
\sum_{k=1}^{t}\E[X_{k}(\crtcl)\mid\calH_{k-1}]
-
\Mart{t}{}{\crtcl}.
\end{align}
Additionally, for any round $t \le \tau$, we have 
$
\Mart{t}{}{\crtcl}=\tildeMart{t}{\crtcl}
$.
\smallskip

Further, for any $k-1 \le \tau$, the $\beta$-stability invariant holds, and this gives
$\delta_{k-1} \ge \beta d_0$.  On the other hand, under $\ereg$, we have  $\Delta_{k-1}  \le 2cd_0$.

Hence by \Cref{lem:exp-matched-neighbours-t} we have,
\begin{align}
\mathbb E[X_k(\crtcl)\mid\calH_{k-1}]
\le
\frac{\Delta_{k-1}}{\delta_{k-1}}
\le
\frac{2c}{\beta}.\label{eq:expec-ht1}
\end{align}

Summing up the bound in \cref{eq:expec-ht1} till round $t$ we get
\begin{align}
\sum_{k=1}^{t}\E[X_k(\crtcl)\mid\calH_{k-1}]
&\le \frac{2c}{\beta} t \label{eq:c3sum1},\\
&\le \frac{2c\alpha}{\beta}d_0 \label{eq:c3sum2}.
\end{align}
Here \cref{eq:c3sum2} follows from \cref{eq:c3sum1} as $t \leq T=\alpha d_0$.

Under $\mathcal E_{\mathrm{reg}}$ we also have that
\[
\deg_0(\crtcl) \ge \frac{c}{2} d_0.
\]
Hence for all $t-1 \le \tau$, we have
\begin{align}
\deg_t(\crtcl)
\ge
\frac{c}{2}d_0
-
\frac{2c\alpha }{\beta}d_0
-
\tildeMart{t}{\crtcl}.\label{eq:deg1comb}
\end{align}
Now suppose $\tau < T$.
Then $\deg_{\tau+1}(\crtcl) < \beta d_0$, and therefore \cref{eq:deg1comb} gives
\begin{align}
\tildeMart{\tau + 1}{\crtcl}
>
\left(\frac{c}{2}-\frac{2c\alpha}{\beta}-\beta\right)d_0 .\label{eq:deg2-comb}
\end{align}

Choose $c$ sufficiently large and  $\alpha, \beta>0$ sufficiently small such that the following holds
\begin{align}
\frac{c}{2}-\frac{2c\alpha}{\beta} -\beta > \frac14 .
\end{align}

Thus \cref{eq:deg2-comb} simplifies and we have,
\begin{align}
\max_{k\le T}\left|\tildeMart{k}{\crtcl}\right| > \frac{d_0}{4}.
\end{align}
By \Cref{lem:freedman-deviation}, we have 

For the critical vertex $\crtcl$,
\begin{align}
\Pr\!\left(
\max_{k\le T}\left|\tildeMart{k}{\crtcl}\right|
> \frac{d_0}{4}
\right)
\le
\exp(-\Theta(\log^{\epsilon} n)).
\end{align}
Since $\eps >0$, we obtain
\begin{align}
\Pr(\tau < T\mid \ereg , \estr,\ejump(T,R))
\le
\exp(-\Theta(\log^{\epsilon} n)) = o(1).
\end{align}
Finally, we have that
\begin{align*}
\Pr(\tau < T)
&\le
\Pr(\mathcal E_{\mathrm{reg}}^c)
+
\Pr(\mathcal E_{\mathrm{struct}}^c)
+
\Pr(\mathcal E_{\mathrm{incr}}^c(T,R))
+
\Pr(\tau < T \mid \mathcal E_{\mathrm{reg}}, \mathcal E_{\mathrm{struct}}, \mathcal E_{\mathrm{incr}}(T,R)) \\
&=
o(1).
\end{align*}
Hence $\tau \geq T$ with probability $1-o(1)$, which implies that
\[
\deg_t(v)\ge \beta d_0
\quad
\text{for all } v\in U_t \text{ and all } t\le T .
\]
This completes the proof.

\end{proof}

This completes the local-to-global argument, showing that strong per-round control of fluctuations lifts to global control of the evolution of vertex degrees, yielding stability for $\Omega(d_0)$ rounds.

%% file: geometric.tex

\newpage

\section{Large Matchings from Invariants}\label{sec:deg-positivity}
In this section we deduce our bound on the size of matching. It uses the weaker degree-positivity invariant, which holds with high probability for the first $\alpha d_0$ rounds by \Cref{thm:persistence}.
Under the degree-positivity invariant, every unmatched vertex has at least
one available neighbour at the beginning of each round. We prove that this ensures that
each vertex has a large probability of becoming matched in a round,
leading to a geometric decay in the number of unmatched vertices. As \bstab~implies degree-positivity, this leads to a matching of linear size after $\alpha d_0$ rounds.

\medskip
\subsection{Geometric Decay Under Sustained Degree-Positivity}\label{sec:geometric-decay}

\begin{lemma}[One-step geometric decay under degree-positivity]
\label{lem:geometric-decay-iterative} Consider an execution of \ref{alg:main} on an input graph $G(V,E)$ with maximum degree at most $\Delta$.

Let $\mathcal{E}_t$ denote the event that the degree-positivity invariant
holds up to round $t$.

Then for every $t\ge1$,
\begin{align}
\E[\abs{U_t} \cdot \mathds{1}_{\mathcal{E}_t} \mid \mathcal{H}_{t-1}]
\le
\abs{U_{t-1}} \cdot \mathds{1}_{\mathcal{E}_{t-1}} \cdot \left(1-\frac{1}{\Delta}\right).\label{eq:geometric1}
\end{align}
Consequently,
\begin{align}
\E[\abs{U_t} \cdot \mathds{1}_{\mathcal{E}_t}]
\le
\E[\abs{U_{t-1}} \cdot \mathds{1}_{\mathcal{E}_{t-1}}]
\left(1-\frac{1}{\Delta}\right).\label{eq:geometric2}
\end{align}
\end{lemma}

\begin{proof}
Recall that $\mathcal{H}_{t-1}$ denote the history of the process up to the
start of round $t$. Conditioning on the $\mathcal{H}_{t-1}$ fixes the residual graph $G_{t-1}$.
\smallskip

Note that $\mathcal{E}_t$ is determined by $\mathcal{H}_{t-1}$, since it only involves the degrees
$\deg_{s-1}$ for $s\le t$. If $\mathcal{E}_t$ does not hold, then the left-hand side of
\cref{eq:geometric1} is zero, and the claim follows trivially. We therefore assume that
$\mathcal{E}_t$ holds for the remainder of the proof; in particular, $\mathcal{E}_{t-1}$ holds as well.
\smallskip

Fix $v\in U_{t-1}$. Given $\mathcal{H}_{t-1}$, the proposal choice of $v$
in round $t$ is independent and uniform over $N_{t-1}(v)$. Thus,
\begin{align}
	\Pr(v \in M_t \mid \mathcal{H}_{t-1})
=
\sum_{u\in N_{t-1}(v)}
\frac{1}{\deg_{t-1}(v)\deg_{t-1}(u)}.
\end{align}

Since $\mathcal{E}_{t}$ holds, $\deg_{t-1}(v)\ge1$ and
$\deg_{t-1}(u)\le\Delta$ for every neighbour $u$. Hence,
\[
\Pr(v \in M_t \mid \mathcal{H}_{t-1})
\ge
\frac{1}{\Delta}.
\]
Therefore,
\[
\Pr(v \in U_t \mid \mathcal{H}_{t-1})
\le
1-\frac{1}{\Delta}.
\]
Summing over all $v\in U_{t-1}$ and noting that
$\mathds{1}_{\mathcal{E}_t}\le\mathds{1}_{\mathcal{E}_{t-1}}$ yields
\begin{align*}
\E[\abs{U_t} \cdot \mathds{1}_{\mathcal{E}_t} \mid \mathcal{H}_{t-1}] 
&= \sum\limits_{v \in U_{t-1}} \Pr(v \in U_t \mid \mathcal{H}_{t-1}) \cdot \mathds{1}_{\mathcal{E}_t}\\
&\le
\abs{U_{t-1}}\cdot\mathds{1}_{\mathcal{E}_{t-1}}
\left(1-\frac{1}{\Delta}\right).
\end{align*}
We have thus shown \cref{eq:geometric1} holds. Taking expectations on both sides yields \cref{eq:geometric2}, proving the lemma.
\end{proof}

We now iterate the one-step bound of \Cref{lem:geometric-decay-iterative}
to obtain a global bound on the matching size after $T$ rounds.
\begin{proposition}[Geometric decay under the degree-positivity invariant]
	\label{lem:geometric-decay} Consider an execution of \ref{alg:main} on a graph $G(V,E)$ with maximum degree at most $\Delta$.

	Let $\mathcal{E}_T$ be the event that the degree-positivity invariant holds up to round $T$. 
	For every $T \ge 1$ we have that, 
	\[
	\E[\abs{M_T}] 
	\ge
	\frac{n}{2}\left(1- \left(1-\frac{1}{\Delta}\right)^T - \Pr(\neg \mathcal{E}_{T})\right).
	\]
\end{proposition}
\begin{proof}
	Let $T \geq 1$.
	We begin by decomposing $\E[\abs{U_T}]$ using conditional expectations,
		\begin{align}
		\E[\abs{U_T}] = \E[\abs{U_T}\cdot \mathds{1}_{\mathcal{E}_T}] + \E[\abs{U_T}\cdot \mathds{1}_{\neg \mathcal{E}_T}].\label{eqn:unmatched-expansion}
	\end{align} 
We will separately bound the two terms in \cref{eqn:unmatched-expansion}. 
\smallskip

Consider the term $ \E[\abs{U_T}\cdot \mathds{1}_{\mathcal{E}_T}]$. By \Cref{lem:geometric-decay-iterative} we have that,
	\begin{align}
		\E[\abs{U_T} \cdot \mathds{1}_{\mathcal{E}_T}]
		\le
		\E[\abs{U_{T-1}} \cdot \mathds{1}_{\mathcal{E}_{T-1}}]
		\left(1-\frac{1}{\Delta}\right). \notag
	\end{align}
	Applying \Cref{lem:geometric-decay-iterative} recursively gives,
	\begin{align}
		\E[\abs{U_T} \cdot \mathds{1}_{\mathcal{E}_T}]
		\le
		\E[\abs{U_{0}} \cdot \mathds{1}_{\mathcal{E}_{0}}]
		\left(1-\frac{1}{\Delta}\right)^{T}. \notag
	\end{align}
	Since $\abs{U_0} = n$, and as $\mathcal{E}_0$ holds trivially with probability $1$, we have,
	\begin{align}
		\E[\abs{U_T} \cdot \mathds{1}_{\mathcal{E}_T}]
	\le
	n\left(1-\frac{1}{\Delta}\right)^T.\label{eqn:unmatched-bound1}
	\end{align}
	\smallskip
Next, we consider the term $\E[\abs{U_T}\cdot \mathds{1}_{\neg \mathcal{E}_T}]$. As	 $\abs{U_T}\leq n$, we have 
\smallskip
\begin{align}
\E[\abs{U_T}\cdot \mathds{1}_{\neg \mathcal{E}_T}]\leq  n\cdot \Pr(\neg \mathcal{E}_T)\label{eqn:unmatched-not}
\end{align}
\smallskip
Plugging in the bounds obtained in \cref{eqn:unmatched-bound1,eqn:unmatched-not} into \cref{eqn:unmatched-expansion} we have,
	\begin{align}
		\E[\abs{U_T}] \leq n\left(1-\frac{1}{\Delta}\right)^T + n\cdot \Pr(\neg \mathcal{E}_T).\label{eqn:unmatched-unconditional}
	\end{align}
	\smallskip
	Since the number of matching edges is half of the number of matched vertices, we have
	\smallskip
	\begin{align}
		\E[\abs{M_T}] 
		= \frac{1}{2}\left(n - \E[\abs{U_T}]\right).\label{eqn:matched-unmatched}
	\end{align}
	\smallskip
	Finally, using \cref{eqn:matched-unmatched} and the bound obtained in \cref{eqn:unmatched-unconditional}, we have
	\[
		\E[\abs{M_T}] 
	\ge
	\frac{n}{2}\left(1- \left(1-\frac{1}{\Delta}\right)^T - \Pr(\neg \mathcal{E}_T)\right),
	\]
	thus proving the lemma.
\end{proof}

\begin{remark}
We note that \Cref{lem:geometric-decay} makes no assumptions about the
existence or structure of a maximum matching in the underlying graph $G$.
In particular, it does not assert that every vertex is matchable,
nor that the algorithm must eventually produce a large matching.
Rather, it states that \text{conditional} on the degree-positivity
invariant holding for $T$ rounds, the probability that any given vertex
remains unmatched decays geometrically with $T$.
\end{remark}
The lemma can therefore be viewed as a general algorithmic principle:
whenever the degree-positivity invariant can be maintained for sufficiently
many rounds, strong guarantees on the size of the resulting matching follow.

\smallskip

Finally, \Cref{lem:geometric-decay} bounds only the \emph{expected} size of the matching after $T$
rounds. To obtain a high-probability bound, we use that the number of edges added in a single round
is concentrated around its conditional expectation (\Cref{lem:app-second-moment} in
\Cref{sec:app-second-moment}), and apply this round by round.
Since $\beta$-stability implies degree positivity, \Cref{thm:persistence} shows that, with high
probability, the degree-positivity invariant holds for the first $T=\alpha d_0$ rounds. Together with
\Cref{lem:geometric-decay-iterative,lem:app-second-moment}, this yields the following theorem.

\begin{theorem}[Large matching after $\Theta(d_0)$ rounds]
\label{thm:linear-rounds-matching}
Fix any $\epsilon>0$. Let $G \sim \mathbb{G}(n,p)$ with $p = \frac{c \log^{1+\epsilon} n}{n}$ for a sufficiently large
constant $c>0$, and let $d_0 := \log^{1+\epsilon} n$. Let $\alpha>0$ be the
constant of \Cref{thm:persistence} and let $T=\alpha d_0$. Then, with probability $1-o(1)$,
\[
|M_T| \ge \frac{n}{2}\left(1-\exp\!\left(-\frac{\alpha}{2c}\right)-o(1)\right).
\]
\end{theorem}
\begin{proof}
Let $\mathcal E_t$ be the event that the degree-positivity invariant holds up to round $t$; it is
determined by $\calH_{t-1}$. On $\ereg$ the maximum degree of $G$ is at most $2pn=2c\,d_0$, so
\Cref{lem:geometric-decay-iterative} with $\Delta=2c\,d_0$ gives, on $\ereg\cap\mathcal E_t$,
\[
\E\bigl[|U_t|\,\big|\,\calH_{t-1}\bigr]\le\Bigl(1-\frac{1}{2c\,d_0}\Bigr)|U_{t-1}|.
\]
Let $\lambda:=\sqrt{n\log n}$. Since $|U_t|=|U_{t-1}|-2(|M_t|-|M_{t-1}|)$, \Cref{lem:app-second-moment} gives, for
every $t$,
\[
\Pr\Bigl(|U_t|>\E\bigl[|U_t|\,\big|\,\calH_{t-1}\bigr]+2\lambda\Bigr)
\le 2\exp\!\left(-\frac{2\lambda^2}{n}\right)=2n^{-2},
\]
and a union bound shows that, with probability at least $1-2Tn^{-2}$, this happens in no round
$t\le T$. On this event intersected with $\ereg\cap\mathcal E_T$, every $t\le T$ satisfies
\[
|U_t|\le\Bigl(1-\frac{1}{2c\,d_0}\Bigr)|U_{t-1}|+2\lambda,
\]
and unrolling from $|U_0|=n$ gives, for $T=\alpha d_0$,
\[
|U_T|\le n\Bigl(1-\frac{1}{2c\,d_0}\Bigr)^{T}+2T\lambda
\le n\,e^{-\alpha/(2c)}+2\alpha d_0\sqrt{n\log n}
= n\bigl(e^{-\alpha/(2c)}+o(1)\bigr).
\]
Finally, $\Pr(\ereg^c)\le n^{-10}$ by \Cref{fact:deg-conc}, and $\Pr(\mathcal E_T^c)=o(1)$ by
\Cref{thm:persistence}, since \bstab~implies degree-positivity. As $|M_T|=(n-|U_T|)/2$, the claimed
bound on $|M_T|$ holds with probability $1-o(1)$.
\end{proof}

\medskip
\section*{Acknowledgements and Disclosure}
We are grateful to Yassine Hamoudi for carefully going through an earlier version of this article.
\medskip

The technical contents of this article, including the proof of the main theorem, were obtained by the authors before April 2026. AI tools were used to improve the writeup and to help locate related literature. The authors take full responsibility for the content of the article.

\newpage

%% file: appendix.tex

\part*{Appendix}

\section{Concentration Inequalities}\label{sec:app-conc-ineq}
In this section we collect the concentration inequalities used throughout the paper. Proofs can be found in any standard textbook on probability theory, e.g. \cite{GrimmettS01}.

\smallskip

\begin{fact}[Union bound]
\label{fact:union-bound}
Let $A_1, A_2, \dots, A_k$ be events. We have
\[
\Pr\!\left(\bigcup_{i=1}^{k} A_i\right)
\le
\sum_{i=1}^{k} \Pr(A_i).
\]
\end{fact}

\smallskip

\begin{fact}[Markov's inequality]
\label{fact:markov}
Let $X$ be a nonnegative random variable and let $a>0$.
Then,
\[
\Pr(X \ge a) \le \frac{\mathbb E[X]}{a}.
\]
\end{fact}

\smallskip

\begin{fact}[Chernoff bound for sums of independent Bernoulli random variables]
\label{fact:chernoff}
Let $X_1,\dots,X_n$ be independent Bernoulli random
variables with $\Pr(X_i=1)=p_i$. Let
\[
X = \sum_{i=1}^n X_i
\qquad \text{and} \qquad
\mu = \mathbb E[X] = \sum_{i=1}^n p_i .
\] 
Then for any $\delta \in (0,1)$,
\[
\Pr(X \le (1-\delta)\mu)
\le
\exp\!\left(-\frac{\delta^2\mu }{2}\right).
\]
Additionally, for any $\delta > 0$,
\[
\Pr(X \ge (1+\delta)\mu)
\le
\exp\!\left(-\frac{\delta^2\mu}{2+\delta}\right).
\]
\end{fact}

\smallskip

\begin{fact}[Azuma--Hoeffding inequality (Lemma A.7 in \cite{CesaL06})]
\label{fact:azuma}
Let $\big(\Mart{t}{}{}\big)_{t = 0}^{n}$ be a martingale with respect to a filtration
$\big(\fil{t}{}{}\big)_{t=0}^{n}$.

Suppose there exist constants $c_1,\dots,c_n$ such that
\[
|\Mart{k}{}{} - \Mart{k-1}{}{}| \le c_k
\quad\text{almost surely for all } k\in [n].
\]
Then for every $\lambda>0$,
\[
\Pr\!\left(|\Mart{n}{}{} - \Mart{0}{}{}| \ge \lambda\right)
\le
2\exp\!\left(-\frac{\lambda^2}{2\sum_{k=1}^n c_k^2}\right).
\]
\end{fact}

\smallskip

\newpage
\begin{fact}[Freedman's Inequality (Lemma A.8 in \cite{CesaL06})]
\label{fact:freedman}
Let $\big(\Mart{t}{}{}\big)_{t \in \N}$ be a real-valued martingale with respect to a
filtration $\big(\mathcal{F}_t\big)_{t \in \N}$, with $\Mart{0}{}{} = 0$. Let
$X_t = \Mart{t}{}{} - \Mart{t-1}{}{}$ denote the martingale increments. Suppose that the increments
are uniformly bounded, i.e.,
\[
\abs{X_t} \le R
\quad \text{almost surely for all } t \in [T].
\]
Define the predictable quadratic variation
\[
V_T = \sum_{t=1}^T \E[X_t^2 \mid \mathcal{F}_{t-1}].
\]
Then for any $\lambda > 0$,
\[
\Pr\!\left(
\max_{t \leq T} \Mart{t}{}{} \ge \lambda \text{ and } V_T \le \sigma^2
\right)
\le
\exp\!\left(
-\frac{\lambda^2}{2\sigma^2 + \frac{2}{3}R\lambda}
\right).
\]
\end{fact}

\smallskip

\begin{fact}[McDiarmid's Inequality (Theorem 20.10 in \cite{Roch24})]
\label{fact:mcdiarmid}
Let $X_1,\dots,X_n$ be independent random variables with 
$X_i \in \mathcal{X}_i$, and let
\[
f : \mathcal{X}_1 \times \cdots \times \mathcal{X}_n \to \mathbb{R}.
\]
Assume that there exist constants $c_1,\dots,c_n$ such that for every 
$i \in [n]$ and all $x_1,\dots,x_n$,
\[
\sup_{x_i' \in \mathcal{X}_i}
\abs{f(x_1,\dots,x_i,\dots,x_n)
-
f(x_1,\dots,x_i',\dots,x_n)}
\le c_i.
\]
Then for every $\varepsilon > 0$,
\[
\Pr\big(
\abs{f(X_1,\dots,X_n) - \E[f(X_1,\dots,X_n)]}
\ge \epsilon
\big)
\le
2\exp\!\left(
-\frac{2\epsilon^2}{\sum_{i=1}^n c_i^2}
\right).
\]
\end{fact}

\smallskip

\begin{fact}[Hoeffding's lemma (Lemma 2.2 in \cite{BoucheronLM13})]
\label{fact:eps-zero-hoeffding}
Let $Y$ be a random variable with $\E[Y]=0$ and $a\le Y\le b$ almost surely. Then for every
$\theta\in\R$,
\[
\E\big[e^{\theta Y}\big]\le\exp\!\left(\frac{\theta^2(b-a)^2}{8}\right).
\]
\end{fact}

\smallskip

\begin{definition}[Submartingale]
    \label{def:submartingale}
    Let $\big(\mathcal{F}_t\big)_{t\in \N}$ be a filtration.
    A sequence of random variables $\big(\Zmart{t}{}{}\big)_{t\in \N}$ is called a
    \emph{submartingale with respect to the filtration $\big(\mathcal{F}_t\big)_{t\in \N}$} if $\E\left[\abs{\Zmart{t}{}{}}\right] < \infty$ for all $t \in \N$, and for every $t \ge 1$, $\E\left[\Zmart{t}{}{} \mid \fil{t-1}{}{}\right] = \Zmart{t-1}{}{}$.
\end{definition}

\begin{fact}[Doob's maximal inequality \cite{GrimmettS01}]
\label{fact:eps-zero-doob}
Let $\big(\mathbf{Q}_t\big)_{t\in\N}$ be a nonnegative submartingale. Then for every $T\in\N$ and every
$a>0$,
\[
\Pr\!\left(\max_{t\le T}\mathbf{Q}_t\ge a\right)\le\frac{\E[\mathbf{Q}_T]}{a}.
\]
\end{fact}

\newpage
\section{Degree Concentration in $\gnp$}
\label{sec:app-degree-conc}
In this section, we prove that the degree-regularity event $\ereg$ from \Cref{def:ereg}
holds with high probability in $\gnp$ via standard Chernoff bounds.
\begin{fact}[Degree concentration in $\mathds G(n,p)$]
\label{fact:deg-conc}
Let $G \sim \mathds G(n,p)$ with
$
p \ge \frac{c\log n}{n}
$
for a sufficiently large constant $c>0$.  
Then
\[
\Pr(\ereg) \ge 1-n^{-10}.
\]
\end{fact}

\begin{proof}
Fix a vertex $x \in V(G)$. Then
\[
\deg_G(x) \sim \mathrm{Bin}(n-1,p),
\qquad
\mu := \E[\deg_G(x)] = (n-1)p .
\]
Since $p \ge \frac{c\log n}{n}$, we have
\[
\mu = (n-1)p = \Theta(\log n).
\]
By the Chernoff bounds (\Cref{fact:chernoff}), for $0<\delta<1$,
\[
\Pr(\deg_G(x) \le (1-\delta)\mu)
\le
\exp\!\left(-\frac{\delta^2\mu}{2}\right),
\]
and for $\delta>0$,
\[
\Pr(\deg_G(x) \ge (1+\delta)\mu)
\le
\exp\!\left(-\frac{\delta^2\mu}{3}\right).
\]
For sufficiently large $n$ we have
\[
\frac{pn}{2} \le (1-\tfrac12)\mu,
\qquad
2pn \ge (1+1)\mu .
\]
Hence
\[
\Pr\!\left(
\deg_G(x) \notin [\tfrac{pn}{2},\,2pn]
\right)
\le
2\exp(-\Theta(\mu)).
\]
Since $\mu=\Theta(\log n)$, this gives
\[
\Pr\!\left(
\deg_G(x) \notin [\tfrac{pn}{2},\,2pn]
\right)
\le
n^{-\Omega(1)} .
\]

Applying a union bound over all $n$ vertices we obtain,
\[
\Pr(\ereg^c)
\le
n \cdot n^{-\Omega(1)}
=
n^{1-\Omega(1)} .
\]

Choosing $c$ sufficiently large ensures
\[
\Pr(\ereg^c) \le n^{-10},
\]
which proves the claim.
\end{proof}

\newpage
\section{Concentration of the Matching Size in One Round}\label{sec:app-second-moment}

\Cref{lem:geometric-decay} bounds the expected number of unmatched vertices. To turn this into a
high-probability bound, we show that, given the history, the number of edges added in a single round
is concentrated around its conditional expectation; the proof of \Cref{thm:linear-rounds-matching}
applies this round by round. A bounded-differences argument over all $T$ rounds at once is not
available: changing one proposal in round $t$ changes which vertices are present in round $t+1$, and
the effect propagates, so already for $T=2$ the matching size can change by $\deg(v)-2$.

\begin{lemma}[Matching size concentration in one round]
\label{lem:app-second-moment}
Let $G$ be any graph and let $t\ge1$. Let $Y_t:=|M_t|-|M_{t-1}|$ be the number of edges added to the
matching in round $t$. Then for every $\lambda>0$,
\[
\Pr\!\left(\bigl|Y_t-\E[Y_t\mid\calH_{t-1}]\bigr|\ge\lambda \;\middle|\; \calH_{t-1}\right)
\le 2\exp\!\left(-\frac{2\lambda^2}{|U_{t-1}|}\right).
\]
\end{lemma}

\begin{proof}
Condition on $\calH_{t-1}$; this fixes the residual graph $G_{t-1}$. Let $U'\subseteq U_{t-1}$ be the
set of unmatched vertices with $\deg_{t-1}(u)\ge1$. In round $t$, each $u\in U'$ proposes to a vertex
$\pi_t(u)$ chosen uniformly from $N_{t-1}(u)$, independently of the other proposals, and
\[
Y_t=\bigl|\bigl\{\{u,w\} : \pi_t(u)=w \text{ and } \pi_t(w)=u\bigr\}\bigr|
\]
is a function of the independent random variables $(\pi_t(u))_{u\in U'}$.

Suppose we change a single proposal $\pi_t(u)$ from $w$ to $w'$ and keep all other proposals fixed.
Only pairs containing $u$ can change: before the change, $u$ is matched if and only if
$\pi_t(w)=u$, and after it, if and only if $\pi_t(w')=u$. Hence $Y_t$ changes by at most $1$.

Applying McDiarmid's inequality (\Cref{fact:mcdiarmid}) to the conditional distribution given
$\calH_{t-1}$, with $c_u=1$ for every $u\in U'$, gives the bound with $|U'|$ in place of $|U_{t-1}|$.
Since $|U'|\le|U_{t-1}|$, the lemma follows. If $U'=\emptyset$, then $Y_t=0$ and there is nothing to
prove.
\end{proof}

\newpage
\section{Proofs for the Warmup: \bstab~via Iterated Weak Concentration}\label{app:iterating-deg-pos}
In this section we provide the deferred proofs of \Cref{lem:additive-loss-t-round} and  \Cref{cor:weak-iterative-bound} from \Cref{sec:warmup}.

\begin{lemma}[Restatement of \Cref{lem:additive-loss-t-round}] 
Fix $\eps \geq 0$. Let $G \sim \mathbb{G}(n,p)$ with $p = \frac{c \log^{1+\eps} n}{n}$ and let $R := 8 + \sqrt{160 c \log^{2+\epsilon} n}$. Suppose that $TR < \frac{c}{4}\log ^{1+\epsilon} n$.

Then for every round $t \le T$ and every vertex $v \in U_t$, we have 
\[
\deg_t(v) \ge \deg_0(v) - tR
\]
with probability at least $1 - (T+1)\cdot n^{-9}$.
\end{lemma}

\begin{proof}

We condition on $\ereg$ for the remainder of this proof. By \Cref{fact:deg-conc}, $\Pr(\ereg) \ge 1 - n^{-9}$.
    
\smallskip

For each round $s \in [T]$ we define the event $E_s$:
\[
E_s :=
\left\{
\deg_s(v) \ge \deg_0(v) - sR
\text{ for all } v \in U_s
\right\}.
\]
We will also need the event $A_t$,
\[
A_t := \bigcap_{s=0}^{t} E_s.
\]
We intend to prove by induction that for every $t \le T$,
\[
\Pr\!\left(
A_t \mid \ereg
\right)
\ge 1 - t n^{-9}.
\]
\textbf{Base case ($t=1$).}
Condition on the regularity event $\ereg$.
Under $\ereg$ we have
\[
\frac{\Delta_0}{\delta_0} \le 4.
\]
Applying \Cref{lem:roundt-degree-stability} with $t=1$ and
$\lambda = \sqrt{160c\log^{2+\epsilon}n}$ gives
\[
\Pr\!\left(
X_1(v) \ge R
\right) \le n^{-10}
\]
for every vertex $v \in V(G)$.

\smallskip

Equivalently,
\[
\Pr\!\left(
\deg_1(v) \le \deg_0(v) - R
\right) \le n^{-10}.
\]

Taking a union bound over all $v \in V(G)$ yields
\[
\Pr(E_1^c \mid \ereg) \le n^{-9}.
\]
Hence
\[
\Pr(A_1 \mid \ereg) \ge 1-n^{-9},
\]
which establishes the base case.

\smallskip

\textbf{Inductive step.}
Assume that for some $1 \le t < T$, the inductive hypothesis holds, i.e.,
\[
\Pr(A_{t-1}^c \mid \ereg) \le (t-1)n^{-9}.
\]
Since $A_t = A_{t-1} \cap E_t$, we have
\[
A_t^c = A_{t-1}^c \cup E_t^c,
\]
and therefore
\begin{align}
\Pr(A_t^c \mid \ereg)
&\le
\Pr(A_{t-1}^c \mid \ereg)
+
\Pr(E_t^c \cap A_{t-1} \mid \ereg). 
\label{eq:induction1}
\end{align}
Fix an arbitrary history $\calH_{t-1}$ for which $A_{t-1}$ holds.
Then for every $v \in U_{t-1}$,
\[
\deg_{t-1}(v) \ge \deg_0(v) - (t-1)R .
\]
Under the regularity event $\ereg$ we have
\[
\frac{c}{2}\log^{1+\epsilon}n
\le
\deg_0(v)
\le
2c\log^{1+\epsilon}n.
\]
Since $TR < \frac{c}{4}\log^{1+\epsilon}n$, this implies
\[
\delta_{t-1}
\ge
\frac{c}{2}\log^{1+\epsilon}n - (t-1)R
\ge
\frac{c}{4}\log^{1+\epsilon}n,
\qquad
\Delta_{t-1} \le 2c\log^{1+\epsilon}n .
\]
By \Cref{lem:exp-matched-neighbours-t} we have
\[
\E\left[ X_t(v) \mid \calH_{t-1} \right]
\le
\frac{\Delta_{t-1}}{\delta_{t-1}}
\le 8.
\]
Applying \Cref{lem:roundt-degree-stability} gives
\[
\Pr\!\left(
X_t(v) \ge 8 + \lambda \mid \calH_{t-1}
\right)
\le
\exp\!\left(
-\frac{\lambda^2}{8\deg_{t-1}(v)}
\right).
\]
Setting
\[
\lambda = \sqrt{160c\log^{2+\epsilon}n}
\]
and using $\deg_{t-1}(v) \le \Delta_{t-1}\le 2c\log^{1+\epsilon}n$ yields
\[
\Pr\!\left(
X_t(v) \ge R \mid \calH_{t-1}
\right) \le n^{-10}.
\]
Taking a union bound over all $v \in U_{t-1}$ gives
\begin{align}
\Pr(E_t^c \mid A_{t-1},\ereg) \le n^{-9}.
\label{eq:induction2}
\end{align}
Substituting \cref{eq:induction2} into \cref{eq:induction1} gives 
\[
\Pr(A_t^c \mid \ereg)
\le
(t-1)n^{-9} + n^{-9}
=
tn^{-9},
\]
thus completing the inductive step.

\smallskip

Finally, since $\Pr(\ereg) \ge 1 - n^{-9}$, we have
\[
\Pr(A_T)
\ge
\Pr(A_T \mid \ereg)\Pr(\ereg)
\ge
\left(1 - T n^{-9}\right)\left(1 - n^{-9}\right)
\ge
1 - (T+1)\cdot n^{-9},
\]
which proves the lemma.
\end{proof}

\newpage

\begin{corollary}
[Restatement of \Cref{cor:weak-iterative-bound}]
Fix $\epsilon\geq 0$. Let $G \sim \mathbb{G}(n,p)$ with $p = \frac{c\log^{1+\epsilon} n}{n}$.
Then the \bstab~invariant holds up to
\[
T = \Theta(\log^{\epsilon/2} n)
\]
rounds with probability at least $1-o(1)$.
\end{corollary}

\begin{proof}
Let $T = \left\lfloor \frac{1}{16}\sqrt{\frac{c}{160}}\,\log^{\epsilon/2} n \right\rfloor$. For sufficiently large $n$, this choice ensures that 
\[
    TR \le \tfrac c4 \log^{1+\epsilon} n.
\]
Hence by \Cref{lem:additive-loss-t-round}, with probability at least
$1-(T+1)n^{-9}=1-o(1)$ we have
\[
\deg_t(v) \ge \deg_0(v)-tR
\qquad
\text{for all } v \text{ and } t\le T .
\]
Under $\ereg$, $\deg_0(v)\ge \tfrac c2 \log^{1+\epsilon} n$, and therefore
\[
\deg_t(v)\ge \tfrac c4 \log^{1+\epsilon} n >0 .
\]
Thus the \bstab~invariant 
holds for all $t\le T$.
\end{proof}